\documentclass[twocolumn,amsmath,amssymb,superscriptaddress,nofootinbib]{revtex4}

\usepackage{amsmath}
\usepackage{amsfonts}
\usepackage{amssymb}
\usepackage{amsthm}
\usepackage{physics}
\usepackage{color}
\usepackage{appendix}
\usepackage{graphicx} 

\usepackage[utf8]{inputenc}

\usepackage{bbm} 
\newcommand{\one}{\mathbbm{1}}

\newcommand{\pr}[1]{\left(#1\right)}
\newcommand{\vrr}[1]{\left|#1\right|}
\newcommand{\cor}[1]{\left[#1\right]}
\newcommand{\cur}[1]{\left\{#1\right\}}
\newcommand{\half}{\frac{1}{2}}
\newcommand{\pd}{\partial}

\newcommand{\Lag}{\mathcal{L}}

\newcommand{\R}{\mathbb{R}}

\newcommand{\C}{\mathbb{C}}

\newcommand{\itero}{\mathrm{s}} 

\newcommand{\subeq}[1]{\begin{subequations}#1\end{subequations}}
\newcommand{\splitt}[1]{\begin{split}#1\end{split}}

\newcommand{\mtrx}[1]{{\renewcommand{\arraystretch}{1.0}\begin{pmatrix}#1\end{pmatrix}}}

\newtheorem{thm}{Theorem}
\newtheorem{dfn}{Definition}
\newtheorem{asm}{Assumption}

\begin{document}
\begin{titlepage}
\title{Which part of the stress-energy tensor gravitates?}

\author{Yigit Yargic}
\email{yyargic@perimeterinstitute.ca}
\affiliation{Perimeter Institute, 31 Caroline St N, Waterloo, Ontario, Canada}
\affiliation{Department of Physics and Astronomy, University of Waterloo, Waterloo, Ontario N2L 3G1, Canada}
\author{Laura Sberna}
\email{lsberna@perimeterinstitute.ca}
\affiliation{Perimeter Institute, 31 Caroline St N, Waterloo, Ontario, Canada}
\affiliation{Department of Physics and Astronomy, University of Waterloo, Waterloo, Ontario N2L 3G1, Canada}
\author{Achim Kempf}
\email{akempf@perimeterinstitute.ca}
\affiliation{Perimeter Institute, 31 Caroline St N, Waterloo, Ontario, Canada}
\affiliation{Department of Physics and Astronomy, University of Waterloo, Waterloo, Ontario N2L 3G1, Canada}
\affiliation{Department of Applied Mathematics, University of Waterloo, Waterloo, Ontario, N2L 3G1, Canada}
\affiliation{Institute for Quantum Computing, University of Waterloo, Waterloo, Ontario, N2L 3G1, Canada}

\begin{abstract}

We consider the possibility that, in the semiclassical Einstein equation for cosmological spacetimes, gravity is sourced by the amount of  stress-energy that is above that of the instantaneous ground state. For this possibility to be consistent, the Bianchi identities must continue to hold. This is nontrivial because it means that the ground state expectation value of the stress-energy tensor must be covariantly conserved in spite of the fact that the ground state is generally a different state at different times. We prove that this consistency condition does hold. As a consequence, we find that the vacuum stress-energy which is above the instantaneous ground state does not renormalize the cosmological constant, as long as the instantaneous ground states and the instantaneous adiabatic vacua exist.

\end{abstract}

\maketitle

\end{titlepage}


\section{Introduction}

In the semiclassical Einstein equation, the cosmological constant receives contri\-bu\-tions from the stress-energy of the vacuum of all matter fields~\cite{Padilla:2015aaa,Martin:2012bt,Burgess:2013ara}. However, the measured value of the cosmological constant is extremely small compared to the scale of each known contribution and one must also expect large contributions from new particles that might emerge at energies higher than probed so far. In addition, and most importantly, the required fine tuning is unstable in the sense that it needs to be repeatedly re-fine tuned as more loops are taken into account in the calculation of the vacuum polarization. This problem of \emph{radiative instability} is an essential part of the cosmological constant problem, see, e.g.,~\cite{Burgess:2013ara,Padilla:2015aaa}.

A key question in this context is which part of the in principle infinite stress-energy of a quantum field actually gravitates and therefore contributes to the right-hand side of the semiclassical Einstein equation.

The answer is straightforward for Minkowski space, where the Poincar\'e sym\-metry singles out the Minkowski vacuum state as a reference state for gravity. The assumption then is that only the stress-energy above the stress-energy of the Minkowski vacuum state is gravitating. In generic curved spacetimes, however, it is not obvious which state could play the role of such a reference state, i.e., a state whose stress-energy is subtracted when determining the amount of stress-energy that actually gravitates. 

At this point it is useful to consider that the search for such a gravitational reference state is logically independent from the search for another important reference state, namely the vacuum state in the sense of a no-particle state. 

In fact, a vacuum state in the sense of an overall no-particle state does not exist on generic spacetimes, due to gravity's ability to parametrically excite modes of matter fields and given the observer dependence of the very notion of particle, as demonstrated, e.g., by the Unruh effect. In the case of Friedmann-Lema\^itre-Robertson-Walker (FLRW) spacetimes, the so-called adiabatic vacuum is in a sense the best approximation to a vacuum state in the sense of a no-particle state for comoving observers, see, e.g.,~\cite{Birrell:1982ix}. 

In the present letter, we are concerned with FLRW spacetimes and we will discuss also the adiabatic vacuum and its role as the reference state with respect to which the particle content in a field is determined. However, our focus here will be on the search for the gravitational reference state with respect to which it is determined how much of the stress-energy of a quantum field contributes to the right-hand side of the semiclassical Einstein equation.

Concretely, we consider the possibility that, at least for FLRW spacetimes, the gravitational reference state at any given time is the ground state $\ket{0_{\mathrm{GS}(t)}}$ of the quantum field's energy density operator $\hat{\rho} \sim \hat{T}_{00}$ at that time. This means that it is the stress-energy above the instantaneous ground state expectation value that acts as the source of gravity on the right-hand side of the semiclassical Einstein equation. Note that our proposal does not rely on any particular criterion for identifying the vacuum state in the sense of a no-particle state.

That this ansatz is consistent with diffeomorphism invariance is non-trivial because the stress-energy of the ground state possesses two time dependencies: Since both $\hat{T}_{\mu\nu}(t,x)$ and $\ket{0_{\mathrm{GS}(t)}}$ are time dependent\footnote{Note that even if we consider the Heisenberg picture, in which states do not evolve with time, the instantaneous ground states have a parametric time dependence.}, we will need to show below that the Bianchi identity $\nabla^\mu \hat{T}_{\mu\nu} = 0$ still yields $\nabla^\mu \bra{0_{\mathrm{GS}(t)}} \hat{T}_{\mu\nu}(t) \ket{0_{\mathrm{GS}(t)}} = 0$.

Also, as is well known, at generic times $t$, the lowest energy state is generally not the adiabatic vacuum state. On one hand, this means that the adiabatic vacuum state is generally an energetically excited state. On the other hand, this means that the energetic ground state is a state with particles from the perspective of the adiabatic vacuum. 

We are here proposing, therefore, a clear distinction between the two basic reference states: the vacuum state, relative to which the particle content is defined, and the gravitational reference state with respect to which the gravitating part of the stress-energy is determined.

While the basic idea that only the excess stress-energy above the energetic ground state gravitates is simple, it has potentially far reaching consequences.
As we will show, the contribution of the subtracted vacuum energy density to the renormalization of the cosmological constant vanishes exactly.

We prove our results non-perturbatively for an arbitrary quantum field theory (QFT) on a cosmological background, while only assuming the existence of ground states and adiabatic vacua. The definition of these states is straightforward for free theories, as we show in the examples of Section \ref{sec:Examples}, but it can become a challenging task for interacting theories. Our results will hold for any QFT for which the ground states and adiabatic vacua can be proven to exist.

Despite the different starting point and physical interpretation, the renormalization resulting from our procedure is consistent with other renormalization methods, such as the Hadamard renormalization~\cite{fulling_1989}. Our proposal does not predict the numerical value of the cosmological constant. But once radiative instability is prevented, the cosmological constant problem is much lessened: Since the subtracted vacuum energy density does not contribute to the cosmological constant, its bare value is not renormalized and can be arbitrary, only constrained by observations. What our proposal does is to protect this value from the renormalization due to the vacuum fluctuations of matter fields, making its value natural in the technical sense. The existence of the cosmological constant per se is of course not a conundrum, as it already arose with the discovery of General Relativity~\cite{Bianchi:2010uw}.

Our paper is organized as follows. In Section~\ref{sec:proposal} we outline our proposal, discussing its diffeomorphism invariance and the role of the physical (adiabatic) vacuum. We conclude the Section discussing renormalization and the interpretation of our results, including the radiative stability of the cosmological constant. This is followed in Section~\ref{sec:Examples} by explicit calculations in three toy models: a free scalar, fermion and massive vector boson. 

In the following, we use the Heisenberg picture of quantum mechanics, in which operators are time-dependent and individual states are constant. We use the mostly-minus signature for the metric and set $\hbar = c = 1$.

\section{The effective stress-energy tensor}
\label{sec:proposal}

We begin by considering the semiclassical Einstein equation
\begin{align} \label{eq:semicl}
\frac{1}{8 \pi G} \, G_{\mu\nu}
- \rho_{\Lambda} \, g_{\mu\nu}
=	{\langle \hat{T}_{\mu\nu} \rangle}_\Psi
\;,
\end{align}
where $\rho_{\Lambda} = \frac{1}{8\pi G} \, \Lambda$ is the energy density corresponding to the cosmological constant. The metric is treated as a classical background field, while matter fields are quantized. We regard~\eqref{eq:semicl} as an effective field theory that arises as the low energy limit of a more fundamental theory of gravity~\cite{Donoghue:1995cz}. We are, therefore, neglecting higher order curvature terms that, in the ultraviolet, might arise on either sides of the equation. 

The expectation value of the stress-energy tensor ${\langle \hat{T}_{\mu\nu} \rangle}_\Psi = \bra{\Psi} \hat{T}_{\mu\nu} \ket{\Psi}$ in the state of the field, $\vert\Psi\rangle$, is a divergent quantity which we assume covariantly regularized. In Minkowski spacetime, we can use the unique ground state of the Hamiltonian, $\ket{0_{\mathrm{M}}}$ as a gravitational reference state. One postulates that it is only the stress-energy above the stress-energy of this state that gravitates: 
\begin{align}
\label{MinkEff}
{\langle \hat{T}_{\mu\nu} \rangle
}_\Psi^{\mathrm{ren.}}
=	{\langle \hat{T}_{\mu\nu} \rangle}_\Psi
	- {\langle \hat{T}_{\mu\nu} \rangle
	}_{\mathrm{M}}
\;.
\end{align}
In Minkowski spacetime, the ground state $\ket{0_{\mathrm{M}}}$ of the Hamiltonian respects the Poincar\'e symmetry of the background, and it also happens to play the role of the vacuum state, in the sense that it is the no-particle state.

As discussed in the introduction, in a fully generic curved spacetime, the notions of vacuum state and gravitational reference state are both highly nontrivial. In this paper, we therefore specialize to homogeneous and isotropic (FLRW) $n$-dimensional spacetimes.
\begin{asm} \label{asm:FLRW}
    The background metric is that of an $n$-dimensional FLRW space\-time.
\end{asm}
First, we note that to preserve the symmetries of FLRW spacetimes, the expectation value ${\langle \hat{T}_{\mu\nu} \rangle}$ has only two distinct non-zero components: the energy density $\rho = n^\mu n^\nu {\langle \hat{T}_{\mu\nu} \rangle}$, where $n^\mu$ is the unit normal to the homogeneous hyper\-surfaces, and the pressure $p$ from the relation ${\langle \hat{T} \rangle} = \rho - \pr{n-1} p$. For most of this paper, we will use the preferred foliation that exists on an FLRW spacetime such that the homogeneous hypersurfaces are parametrized by the time coordinate $t \equiv x^0$ and $n^\mu = \sqrt{g^{00}} \, (1,0,...,0)$.

We now define an \emph{effective}\footnote{The reason for calling this quantity \emph{effective} instead of \emph{renormalized} will be explained in Section \ref{sec:Discussion}.} expectation value of the stress-energy tensor by subtracting the expectation value of the stress-energy of a \emph{gravitational reference state}, similar to the case of a Minkowski spacetime in \eqref{MinkEff}, in a way that also ensures that the resulting \emph{effective} energy density that actually gravitates is positive.

In a sense, the obvious choice for this subtraction is to subtract the stress-energy of the state with lowest energy density. And in homogeneous spacetimes, minimizing the energy density is of course equivalent to minimizing the total energy, at any given time. We call the states under consideration \emph{ground states}\footnote{In the literature, the term ground state is reserved for the eigenstate of the Hamiltonian operator with lowest eigenvalue, which might be different from the state with lowest energy density that we consider here. This distinction is particularly important if the quantum field is coupled to the curvature. The fact that we minimize the energy density will turn out to be important to preserve the Bianchi identity.} and we define them formally as follows.
\begin{dfn} \label{dfn:GS}
Given a quantum field theory (QFT) on an FLRW spacetime, the \textbf{(instantaneous) ground state at time $t$} is defined as the vacuum state $\ket{0_{\mathrm{GS}(t)}}$ which respects the symmetries of the background metric and satisfies
\begin{align}
\label{DefPropGS}
    \delta_\psi \pr{ n^\mu n^\nu
    \bra{\psi} \hat{T}_{\mu\nu}(t) \ket{\psi} }
    {\Big\vert}_{t ;\, \psi = 0_{\mathrm{GS}(t)}}
    = 0
    \;,
\end{align}
where the domain of variation $\delta_\psi$ is the set of all vacuum states that respect the symmetries of the FLRW background.
\end{dfn}
We assume the existence of ground states in the following and we discuss this assumption at the end of Section \ref{sec:Discussion}.
\begin{asm} \label{asm:GS-exists}
    An instantaneous ground state exists uniquely at each moment of time.
\end{asm}
The immediate concern in subtracting the stress-energy of the ground state is, however, that generically no single state minimizes the energy density $\rho(t)$ at all times simultaneously. The ground state at one time $t_0$ is an energetically excited state at another time $t_1$. One can only speak of an \emph{instantaneous} ground state $\ket{0_{\mathrm{GS}(t)}}$ that minimizes the energy at a given time $t$, while the ground states at different times are generically different states~--~see the explicit examples of ground state vacua provided in Section \ref{sec:Examples}.

\textbf{What we are proposing,} therefore, is to define for each time, $t$, its own gravitational reference state, namely the energetic ground state at that time $t$. Using this family $\cur{\ket{0_{\mathrm{GS}(t)}}, t \in \R}$ of instantaneous ground states, we can define the \emph{effective} expectation value of the stress-energy tensor by subtracting the corresponding ground state value at every time, such that
\begin{align}
\label{GSEff}
{\langle \hat{T}_{\mu\nu} \rangle
	}_\Psi^{\mathrm{eff.}} (t)
\equiv
    {\langle \hat{T}_{\mu\nu} (t) \rangle
	}_\Psi
	- {\langle \hat{T}_{\mu\nu} (t) \rangle
	}_{\mathrm{GS}(t)}
\;.
\end{align}
{The counter-intuitive fact that this quantity is covariantly conserved is one of the main results of this paper and it will be proven in the next section as Theorem \ref{thm:magic}}. We then propose that the source to the semiclassical Einstein equation for an FLRW background is the effective part \eqref{GSEff} of the stress-energy tensor expectation value,
\begin{align}
\label{EinsteinEff}
\boxed{
\frac{1}{8 \pi G} \, G_{\mu\nu}
- \rho_{\Lambda} \, g_{\mu\nu}
=	{\langle \hat{T}_{\mu\nu} \rangle
	}_\Psi^{\mathrm{eff.}}
}\;.
\end{align}

\subsection{Covariant conservation law}

When we choose a time-dependent vacuum family for the subtraction as in \eqref{GSEff}, the major concern is preserving the consistency of the semiclassical Einstein equation under a covariant derivative. The left-hand side in \eqref{EinsteinEff} consists of covariantly conserved tensors. For any individual state, the expectation value $\langle \hat{T}_{\mu\nu} \rangle$ of the stress-energy tensor is also covariantly conserved by diffeomorphism invariance. However, this argument does not apply to the vacuum family expectation value ${\langle \hat{T}_{\mu\nu} \rangle}_{\mathrm{GS}(t)}$ because of the parametric time dependence of the state. The reader might expect the conservation law to be broken for this quantity.


This turns out not to be the case. We find that this is a non-trivial property of the ground state family on an FLRW background and summarize this result in the following theorem.
\begin{thm} \label{thm:magic}
The ground state family expectation value of the stress-energy tensor is covariantly conserved on an FLRW spacetime,
\begin{align}
\nabla^\mu {\langle \hat{T}_{\mu\nu} (t) \rangle
	}_{\mathrm{GS}(t)}
	= 0
\;.
\end{align}
\end{thm}
\begin{proof}
The line element on an $n$-dimensional FLRW spacetime is given by
\begin{align}
\dd s^2
=	\dd t^2 - {a(t)}^2 \, \dd \Sigma_{n-1}^2
\;.
\end{align}
where $\dd \Sigma_{n-1}^2$ is the line element on each spatial section. For any time-parametrized family $\cur{\ket{\psi_t}, t \in \R}$ of states which respect the homogeneity of the FLRW back\-ground, such as the ground states, we can write
\begin{align}
    \bra{\psi_t} \hat{T}_{\mu\nu}(t) \ket{\psi_t}
    \dd x^\mu \, \dd x^\nu
&=  \rho(t; \psi_t) \, \dd t^2
    \nonumber \\ &\hspace{0.5cm}
    + p(t; \psi_t) \, {a(t)}^2 \, \dd \Sigma_{n-1}^2
    \;,
\end{align}
where $\rho(t;\psi_t)$ is the energy density and $p(t;\psi_t)$ is the pressure at time $t$ for each state $\ket{\psi_t}$. Writing the components of the covariant derivative explicitly for the ground state expectation values, we find
\begin{align}
\label{Derv1}
\nabla^\mu {\langle \hat{T}_{\mu 0} (t) \rangle
	}_{\mathrm{GS}(t)}
&=	\pr{n-1} \frac{a'(t)}{a(t)}
	\pr{\rho(t; 0_{\mathrm{GS}(t)}) + p(t; 0_{\mathrm{GS}(t)})}
    \nonumber \\ &\hspace{0.5cm}
	+ \frac{\pd}{\pd t} \, \rho(t; 0_{\mathrm{GS}(t)})
\;, \\ \nonumber
\nabla^\mu {\langle \hat{T}_{\mu j} (t) \rangle
	}_{\mathrm{GS}(t)}
&=	0
\;.
\end{align}
Since the spatial components of the covariant derivative vanish identically, we will focus on the time component.

There are two kinds of time-dependence in $\rho(t; 0_{\mathrm{GS}(t)})$. Firstly, for each fixed ground state $\ket{0_{\mathrm{GS}(u)}}$ at time $u \in \R$, the expectation value $\rho(t;0_{\mathrm{GS}(u)}) \equiv {\langle \hat{T}_{00}(t) \rangle}_{\mathrm{GS}(u)}$ is a function of time $t$, since the operator $\hat{T}_{00}(t)$ evolves over time. The second one is the choice of the parameter $u$ that specifies the time at which the state minimizes the energy density. Our proposal sets these two parameters equal, $u = t$. The same discussion applies to $p(t;0_{\mathrm{GS}(u)})$.

Let's distinguish between the parameters $t$ and $u$ for a moment and define $\theta = t - u$. The quantity $\rho(t;0_{\mathrm{GS}(u)})$ depends on two of these variables independently. The first term in \eqref{Derv1} is understood as first setting $u = t$ and then taking the derivative with respect to $t$. This is equivalent to taking the partial derivative with respect to $t$ while holding $\theta$ fixed and then setting $u=t$,
\begin{align}
\label{Derv2}
&\nabla^\mu {\langle \hat{T}_{\mu 0} (t) \rangle
	}_{\mathrm{GS}(t)}
=	\bigg(
	\frac{\pd \rho(t;0_{\mathrm{GS}(u)})}{\pd t} \,
	{\bigg\vert}_{\theta}
    \nonumber \\ &\hspace{0.5cm}
	+ \pr{n-1} \frac{a'(t)}{a(t)}
	\pr{ \rho(t;0_{\mathrm{GS}(u)})
	    + p(t;0_{\mathrm{GS}(u)}) }
	\bigg) {\bigg\vert}_{u=t}
\;.
\end{align}
On the other hand, the standard conservation law gives $\nabla^\mu {\langle \hat{T}_{\mu 0}(t) \rangle}_{\mathrm{GS}(u)} = 0$ for every fixed state, i.e., for every fixed $u$. If we first evaluate the covariant derivative for fixed $u$ and then set $u=t$, this becomes
\begin{align}
\label{Derv3}
0
&=	\bigg(
	\frac{\pd \rho(t;0_{\mathrm{GS}(u)})}{\pd t} \,
	{\bigg\vert}_{u}
    \nonumber \\ &\hspace{0.5cm}
	+ \pr{n-1} \frac{a'(t)}{a(t)}
	\pr{ \rho(t;0_{\mathrm{GS}(u)})
	    + p(t;0_{\mathrm{GS}(u)}) }
	\bigg) {\bigg\vert}_{u=t}
\;.
\end{align}
A simple calculation for the derivatives on the $t$-$u$-space shows that
\begin{align}
\label{Derv4}
\frac{\pd f(t,u)}{\pd t} {\bigg\vert}_{\theta}
= \frac{\pd f(t,u)}{\pd t} {\bigg\vert}_{u}
+ \frac{\pd f(t,u)}{\pd u} {\bigg\vert}_{t}
\end{align}
for every scalar function $f$. Hence, if we subtract \eqref{Derv2} from \eqref{Derv3} and use \eqref{Derv4}, we get
\begin{align}
\label{Derv5}
\nabla^\mu {\langle \hat{T}_{\mu 0} (t) \rangle
	}_{\mathrm{GS}(t)}
&=	
    \frac{\pd \rho(t;0_{\mathrm{GS}(u)})}{\pd u}
	{\bigg\vert}_{t ;\, u=t}
\;.
\end{align}
Note that we have not used any properties of the ground states up to this point; thus~\eqref{Derv5} holds for any time-dependent family of states. The defining property of an instantaneous ground state $\ket{0_{\mathrm{GS}(t)}}$ at time $t$ is that it minimizes the energy density $\rho(t)$ at that time among all states as in \eqref{DefPropGS}.
This implies that $\ket{0_{\mathrm{GS}(t)}}$ also minimizes the instantaneous energy density $\rho(t)$ among the family $\cur{\ket{0_{\mathrm{GS}(u)}}, u \in \R}$ of ground states at different times. Therefore, the right-hand side of \eqref{Derv5} vanishes.
\qedhere
\end{proof}
In conclusion, Theorem \ref{thm:magic} ensures that our proposal for the semiclassical Einstein equation in \eqref{GSEff} and \eqref{EinsteinEff} is consistent with diffeomorphism invariance, i.e., with the Bianchi identity.

Note that the proof of Theorem \ref{thm:magic} relies on diffeomorphism invariance for fixed states, since we use \eqref{Derv3} to convert the time derivative at step \eqref{Derv2} into a derivative over the state parameter at step \eqref{Derv5}. Therefore, the statement of Theorem \ref{thm:magic} does not have an analog in systems that do not possess diffeomorphism invariance, such as the time-dependent harmonic oscillator.

\subsection{Vacua in cosmology}
\label{sec:split}

Naively, one might expect that the momentary ground state is the vacuum state, i.e., the no-particle state. However, in this case, the predicted amount of particle creation would exceed the upper bound from astrophysical observations (see~\cite[p.~73]{Birrell:1982ix} and references therein). As a result, the instantaneous ground state and the related \emph{Hamiltonian diagonalization} were ruled out as vacuum identification criteria. The lowest energy state and the physical no-particle state must be, therefore, distinct states for quantum field theories in generic curved backgrounds.

The candidates for the physical vacuum states that are generally considered to be most plausible for cosmological backgrounds belong to the family of \emph{adiabatic vacuum states}. These vacua are obtained by solving a certain perturbative expansion up to a finite number of derivatives of the metric components. See, e.g., \cite{Luders1990} for a rigorous definition of the concept.
Originally, the adiabatic vacuum states were introduced in~\cite{Parker:1974qw}. 

The basic idea is that the criterion for identifying which state is the vacuum state at any given time, $t$, should be such that the amount of cosmological particle creation which is predicted as a consequence of applying this criterion is minimized.

At this point, we remark that the adiabatic vacuum identification criterion may also be motivated in a new way that is based on first-principles, i.e., without the need to appeal to data. To this end, we begin with the intuition that, when the universe either expands or shrinks, any field state which possesses a nonzero particle content must in some way change over time, the change being due to the fact that the particle content of the state has to either dilute or concentrate. This yields a criterion for identifying the no-particle state. Namely, whatever the criterion for singling out the vacuum state at time $t$ is, it should be such that, when applied over a range of times, the so-obtained vacuum states, parametrized by $t$, should change as little as possible - as determined via Bogolubov $\beta$ coefficients. This then implies the conventional adiabatic vacuum identification criterion: the amount of particle production should be minimal. In this new way, the vacuum is identified as the state that is most immune to dilution and concentration, so that all particle creation or annihilation that does happen due to expansion or shrinkage is solely due to quantum parametric excitation.

Technically, the procedure for identifying the adiabatic vacua is to solve the Wronskian condition\footnote{The Wronskian condition is obtained from the consistency requirement between the canonical commutation relations of the field operators and those of the annihilation and creation operators. More details can be found in the examples in Section \ref{sec:Examples}.} by a WKB-type ansatz and to approach a solution of the equation of motion iteratively around a Minkowski-like 0-th order solution. A finite number $\itero$ of iterations gives an approximate solution failing to be exact only by terms with at least $2\itero+2$ derivatives of the metric. The mode functions for the adiabatic vacua are then found by evaluating the initial conditions for the equation of motion along the approximate solutions. Hence, similarly to the ground state family, the adiabatic vacuum states $\cur{\ket{0_{\mathrm{AV}(t)}}, t \in \R}$ are parametrized by the time at which they are defined. Explicit examples of adiabatic vacuum solutions are provided in Section \ref{sec:Examples}.

The rigorous definition of adiabatic vacua in the literature has focused mostly on free theories and might become non-trivial for generic QFTs. While we did not need to specify a physical vacuum state for our main proposal \eqref{EinsteinEff}, we will rely on this concept in the following for discussing renormalization. Therefore, we include it here as our final assumption.
\begin{asm} \label{asm:AV-exists}
    The adiabatic vacua exist at every time and represent the right choice for physical vacuum states.
\end{asm}
Subtracting the physical~--~here taken to be the adiabatic~--~vacuum family expectation value should already give a renormalized (finite) stress-energy tensor, corresponding to the quantum and classical sources from observed fields,
\begin{align}
\label{AVren}
{\langle \hat{T}_{\mu\nu} \rangle
	}_\Psi^{\mathrm{ren.}} (t)
\equiv
    {\langle \hat{T}_{\mu\nu} (t) \rangle
	}_\Psi
	- {\langle \hat{T}_{\mu\nu} (t) \rangle
	}_{\mathrm{AV}(t)}
\;.
\end{align}
Clearly, this is different from the \emph{effective} part defined in \eqref{GSEff}. We can say that, while the effective part measures the gravitating stress-energy excitation, the renormalized part measures the stress-energy in particle excitation over the adiabatic vacuum. The difference between the two,
\begin{align}
\label{AVvac}
{\langle \hat{T}_{\mu\nu} \rangle
	}^{\mathrm{vac.}} (t)
&\equiv
    {\langle \hat{T}_{\mu\nu} \rangle
	}_\Psi^{\mathrm{eff.}} (t)
	- {\langle \hat{T}_{\mu\nu} \rangle
	}_\Psi^{\mathrm{ren.}} (t)
\nonumber \\ &=
    {\langle \hat{T}_{\mu\nu} (t) \rangle
	}_{\mathrm{AV}(t)}
	- {\langle \hat{T}_{\mu\nu} (t) \rangle
	}_{\mathrm{GS}(t)}
\;,
\end{align}
is independent of the particle content, i.e., it is a purely geometrical contribution. It will be discussed in more detail in Section~\ref{sec:III}.
In conclusion, we can write the effective expectation value as
\begin{align}
\label{AVsplit}
{\langle \hat{T}_{\mu\nu} \rangle
	}_\Psi^{\mathrm{eff.}}
=	{\langle \hat{T}_{\mu\nu} \rangle
	}_\Psi^{\mathrm{ren.}}
	+ {\langle \hat{T}_{\mu\nu} \rangle
	}^{\mathrm{vac.}}~.
\end{align}
See also the schematic representation in Figure~\ref{fig:GS-energy}.
The first term, ${\langle \hat{T}_{\mu\nu} \rangle}_\Psi^{\mathrm{ren.}}$, should be finite and match the observed sources of gravitation, while the second term ${\langle \hat{T}_{\mu\nu} \rangle}^{\mathrm{vac.}}$ measures the elevation of the vacuum energy above the ground state. After making this split in~\eqref{GSEff}, we write the semiclassical Einstein equation (for an FLRW spacetime) as
\begin{align}
\label{EinsteinReno}
\frac{1}{8 \pi G} \, G_{\mu\nu}
- \rho_{\Lambda} \, g_{\mu\nu}
- {\langle \hat{T}_{\mu\nu} \rangle
	}^{\mathrm{vac.}}
=	{\langle \hat{T}_{\mu\nu} \rangle
	}_\Psi^{\mathrm{ren.}}
\;.
\end{align}
The vacuum part ${\langle \hat{T}_{\mu\nu} \rangle}^{\mathrm{vac.}}$ is still divergent and acts as a counter-term in this equation.

\begin{figure}
\begin{center}
\includegraphics[width = 0.4 \textwidth]{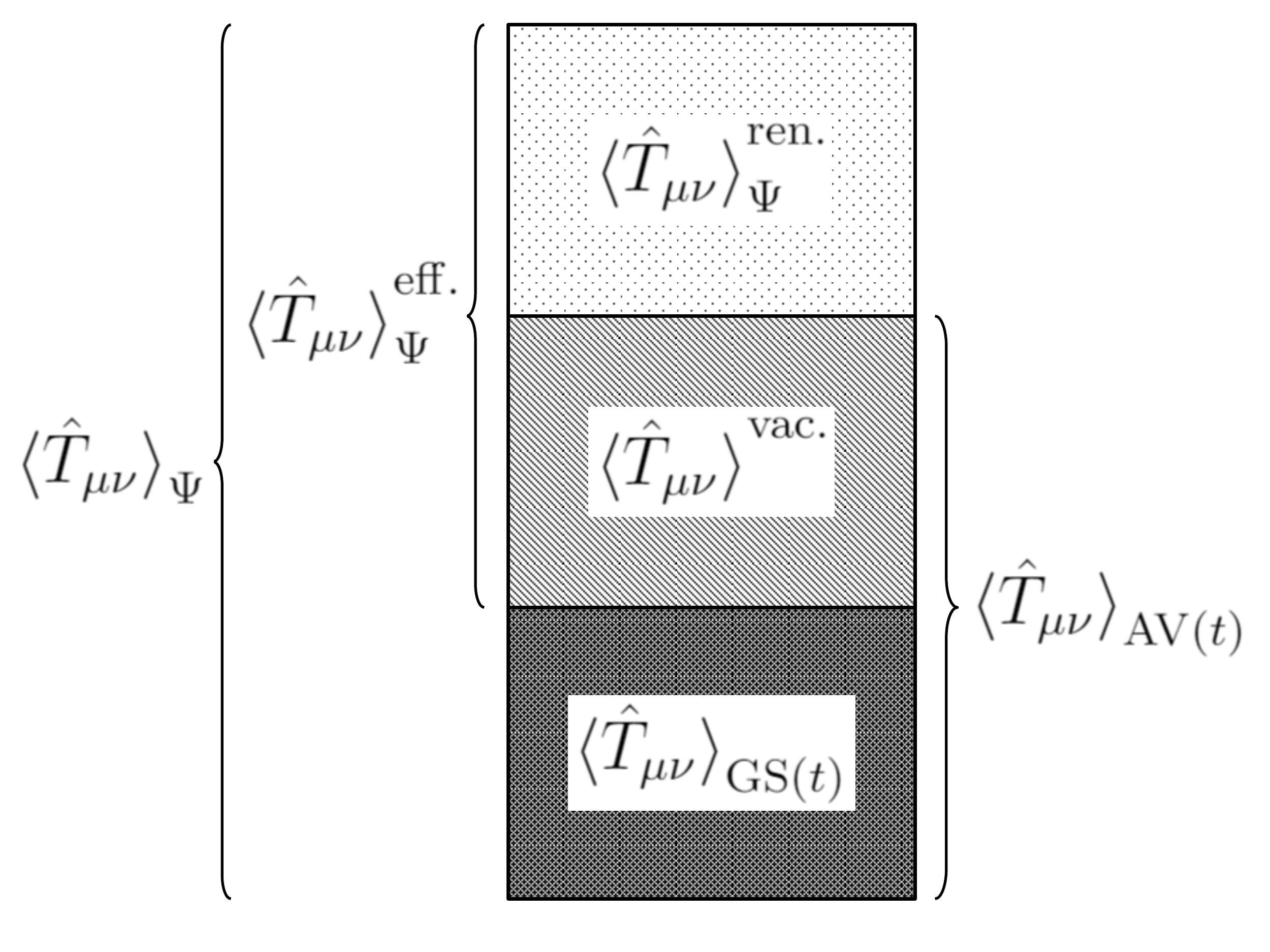}
\end{center}
\caption{The total expectation value ${\langle \hat{T}_{\mu\nu} \rangle}_\Psi$ of the stress-energy tensor are split into various parts based on its ground state value and the (adiabatic) vacuum value. Note that ${\langle \hat{T}_{\mu\nu} \rangle}^{\mathrm{vac.}}$ and ${\langle \hat{T}_{\mu\nu} \rangle}_{\mathrm{GS}(t)}$ are infinite quantities.}
\label{fig:GS-energy}
\end{figure}

Let us further comment on our assumption that the adiabatic vacuum family and the state of the matter fields, $\vert\Psi\rangle$, are such that ${\langle \hat{T}_{\mu\nu} \rangle}_\Psi^{\mathrm{ren.}}$ is finite. In fact, the divergence in the $\vert\Psi\rangle$ expectation value is the same as the divergence in its corresponding no-particle state~\cite{Birrell:1982ix}. Therefore, we are merely assuming that the particles described by $\Psi$ are excitations of the adiabatic vacuum, i.e., that the adiabatic vacuum is the physical no-particle state.

\subsection{Renormalization}\label{sec:III}\label{sec:SubtractionScheme}

Both ${\langle \hat{T}_{\mu\nu} \rangle}_{\mathrm{GS}(t)}$ and ${\langle \hat{T}_{\mu\nu} \rangle}_{\mathrm{AV}(t)}$ only depend on the metric, i.e., they are purely geometric quantities. The former is also covariantly conserved by Theorem \ref{thm:magic}. Since the only covariantly conserved geometric tensors up to second order are the metric $g_{\mu\nu}$ and the Einstein tensor $G_{\mu\nu}$, we can write
\begin{align}
\label{ExpppGS}
{\langle \hat{T}_{\mu\nu}
	\rangle}_{\mathrm{GS}(t)}
=	\mathcal{A}_{\mathrm{GS}} \, g_{\mu\nu}
	+ \mathcal{B}_{\mathrm{GS}} \, G_{\mu\nu}
	+ \text{higher curv.~terms}
\;.
\end{align}
The quantity ${\langle \hat{T}_{\mu\nu} \rangle}_{\mathrm{AV}(t)}$, on the other hand, is \emph{not} covariantly conserved. Never\-the\-less, we prove the following result:
\begin{thm}
\label{thm:AV-conservation}
For adiabatic vacua of at least 2nd order (at least $\itero \geq 1$ iteration), the adiabatic vacuum family expectation value of the stress-energy tensor ceases to be covariantly conserved only by terms that contain at least 4 derivatives of the metric. Schematically, we write this as
\begin{align}
    \nabla^\mu {\langle \hat{T}_{\mu\nu} \rangle}_{\mathrm{AV}(t)}
    = 0 + \mathcal{O}(\pd^4)
    \;.
\end{align}
\end{thm}
\begin{proof}
Similarly to the proof of Theorem \ref{thm:magic}, we distinguish between the time parameter $t$ on which the operator $\hat{T}_{\mu\nu}$ depends and the time parameter $u$ at which an adiabatic vacuum state is defined, before setting $\theta \equiv t - u = 0$. Using again diffeomorphism invariance and the relation \eqref{Derv4} between derivatives, we obtain
\begin{align}
    \nabla^\mu {\langle \hat{T}_{\mu\nu} \rangle}_{\mathrm{AV}(t)}
    &=  \nabla_\mu
        {\langle \hat{T}^\mu{}_{\nu}(t)
        \rangle}_{\mathrm{AV}(u)} \,
        {\Big\vert}_{\theta ;\, u=t}
    \nonumber \\
    &=  \bigg(
        \nabla_\mu
        {\langle \hat{T}^\mu{}_{\nu}(t)
        \rangle}_{\mathrm{AV}(u)} \,
        {\Big\vert}_{u}
    \nonumber \\ &\hspace{0.8cm}
        + \frac{\pd}{\pd u} \,
	    {\langle \hat{T}^0{}_{\nu}(t)
        \rangle}_{\mathrm{AV}(u)} \,
        {\bigg\vert}_{t}
        \bigg) {\bigg\vert}_{u=t}
    \nonumber \\
    &=  \frac{\pd}{\pd u} \,
	    {\langle \hat{T}^0{}_{\nu}(t)
        \rangle}_{\mathrm{AV}(u)} \,
        {\bigg\vert}_{t ;\, u=t}
    \;.
\end{align}
Now, recall that the approximate solution at $\itero$-th iteration, which is used to define the initial conditions for the adiabatic vacua, ceases to be an exact vacuum solution only at the order $\mathcal{O}(\pd^{2\itero+2})$. Hence, the adiabatic vacuum at $u + \delta u$ differs from the one at $u$ at the equal time $t=u$ only by $\mathcal{O}(\pd^{2\itero+2}) \, \delta u + \mathcal{O}(\delta u^2)$.
Then, we find
\begin{align}
    \nabla^\mu {\langle \hat{T}_{\mu\nu} \rangle}_{\mathrm{AV}(t)}
    &= \mathcal{O}(\pd^{2\itero+2})
    \;.
\end{align}
If we consider the adiabatic vacua at iteration $\itero \geq 1$, the statement of the theorem is proven. \qedhere
\end{proof}
Using Theorem \ref{thm:AV-conservation}, we can also write
\begin{align}
\label{ExpppAV}
{\langle \hat{T}_{\mu\nu}
	\rangle}_{\mathrm{AV}(t)}
=	\mathcal{A}_{\mathrm{AV}} \, g_{\mu\nu}
	+ \mathcal{B}_{\mathrm{AV}} \, G_{\mu\nu}
	+ \text{higher curv.}
\end{align}
for the adiabatic vacuum family, because the failure of conservation is completely contained in the higher curvature terms.

The constant coefficients $\mathcal{A}_{\mathrm{GS}}$, $\mathcal{B}_{\mathrm{GS}}$, $\mathcal{A}_{\mathrm{AV}}$ and $\mathcal{B}_{\mathrm{AV}}$ are independent of the background geometry, i.e., they can only depend on the parameters of the theory, such as the field mass.
\begin{thm}
\label{thm:X-property}
The instantaneous ground states and the adiabatic vacua coincide at the zero-th adiabatic order, i.e., $\mathcal{A}_{\mathrm{GS}} = \mathcal{A}_{\mathrm{AV}}$.
\end{thm}
\begin{proof}
We consider smoothly flattening the FLRW background as follows: We replace the scale factor $a(t)$ with $a(t_0 + \epsilon(t- t_0))$ for arbitrary $t_0 \in \R$ and shift the value of $\epsilon$ from $1$ to $0$. The coefficients $\mathcal{A}_{\mathrm{GS}}$, $\mathcal{B}_{\mathrm{GS}}$, $\mathcal{A}_{\mathrm{AV}}$ and $\mathcal{B}_{\mathrm{AV}}$ are independent of $\epsilon$. When $\epsilon$ reaches $0$, the spacetime becomes flat and all curvature terms in \eqref{ExpppGS} and \eqref{ExpppAV} vanish.

Note that the Minkowski vacuum serves as both the ground state and the adiabatic vacuum on a Minkowski spacetime. Therefore, the expectation values ${\langle \hat{T}_{\mu\nu}\rangle}_{\mathrm{GS}(t)}$ and ${\langle \hat{T}_{\mu\nu}\rangle}_{\mathrm{AV}(t)}$ coincide with ${\langle \hat{T}_{\mu\nu}\rangle}_{\mathrm{M}}$ in the flat limit $\epsilon \rightarrow 0$. Since they also converge to $\mathcal{A}_{\mathrm{GS}} \, g_{\mu\nu}$ and $\mathcal{A}_{\mathrm{AV}} \, g_{\mu\nu}$, respectively, and the coefficients are unchanged by the flattening, we conclude $\mathcal{A}_{\mathrm{GS}} = \mathcal{A}_{\mathrm{AV}}$.\qedhere
\end{proof}

By Theorem \ref{thm:X-property}, \eqref{AVvac} becomes
\begin{align}
\label{Vacco}
{\langle \hat{T}_{\mu\nu} \rangle
	}^{\mathrm{vac.}}
	=	\pr{\mathcal{B}_{\mathrm{AV}}
		- \mathcal{B}_{\mathrm{GS}}} G_{\mu\nu}
	+ \text{higher curv.}
\end{align}
Then, we may define the renormalized couplings for the semiclassical Einstein equation \eqref{EinsteinReno} as
\begin{align}
\label{ParaReno}
\rho_{\Lambda}^{\mathrm{ren.}}
=	\rho_{\Lambda}
\quad \text{and} \quad
G^{\mathrm{ren.}}
=	\frac{G}{1 - 8 \pi G
		\pr{\mathcal{B}_{\mathrm{AV}}
		- \mathcal{B}_{\mathrm{GS}}} }
\;.
\end{align}
In accordance with the truncation of the Einstein-Hilbert action, we neglect the higher curvature terms in \eqref{Vacco}. After the renormalization of the parameters, the semiclassical Einstein equation \eqref{EinsteinReno} can finally be written as:
\begin{align}
\label{EinsteinReno2}
\frac{1}{8 \pi G^{\mathrm{ren.}}} \, G_{\mu\nu}
- \rho_{\Lambda}^{\mathrm{ren.}} \, g_{\mu\nu}
=	{\langle \hat{T}_{\mu\nu} \rangle
	}_\Psi^{\mathrm{ren.}}
\;.
\end{align}

By neglecting higher curvature terms, we are neglecting two features of the equation: Firstly, higher curvature terms would contribute to the renormalization of the parameters of a higher curvature gravity theory, beyond the Einstein-Hilbert action. This problem is related to the perturbative non-renormalizability of gravity and is beyond the scope of this paper.

Secondly, the right-hand side of \eqref{EinsteinReno2} fails to be covariantly conserved at higher-than-second adiabatic order, i.e., at the neglected higher curvature terms. This slight violation of the Bianchi identity in the finite part is due entirely to the choice of adiabatic vacua to separate the finite part of the stress-energy tensor from the vacuum contribution. It is neither a consequence of our main proposal in \eqref{EinsteinEff}, nor does it change the fact that, due to Theorem 1, the stress-energy that sources gravity strictly obeys the Bianchi identity.

\subsection{Discussion}
\label{sec:Discussion}

We have argued that the stress-energy expectation of the instantaneous ground state, ${\langle \hat{T}_{\mu\nu} \rangle}_{\mathrm{GS}(t)}$, is a good choice for subtraction from the full stress-energy of the actual state of the matter fields for two reasons: Firstly, this subtraction as in \eqref{GSEff} maintains a positive energy density for the difference, consistent with the weak energy condition. Secondly, ${\langle \hat{T}_{\mu\nu} \rangle}_{\mathrm{GS}(t)}$ is covariantly conserved exactly, as we have shown in Theorem \ref{thm:magic}.

The proposed scheme has direct consequences for the renor\-ma\-li\-zation of the cosmological constant. Recall that the counter-term in~\eqref{EinsteinReno} is a difference of the stress-energy tensor expectation values in the physical (adiabatic) vacuum and in the ground state. As we have shown in Section \ref{sec:SubtractionScheme}, and in \eqref{ParaReno} in particular, this counter-term does not affect the value of the parameter $\rho_\Lambda = \frac{1}{8\pi G} \, \Lambda$. This particular combination of parameters that is linearly related to the cosmological constant becomes protected from vacuum fluctuations. In this sense, our proposal contributes to the resolution of the cosmological constant problem under the given assumptions.

We assumed that the Universe is well described by an FRLW metric at cosmological scales, and that instantaneous ground states exist. Furthermore, we assumed the existence of adiabatic vacua as physical vacua in discussing the renormalization of the Einstein equation. In Section~\ref{sec:Examples}, we analyze three free QFTs, where Assumptions \ref{asm:GS-exists} and \ref{asm:AV-exists} are shown explicitly to hold. For interacting theories the definition and existence of ground states and adiabatic vacua might be more challenging. In specific theories, the one-loop order can be worked out following the techniques in \cite{Markkanen:2013nwa}.

Free theories still provide important insights on the renormalization of the cosmological constant at different scales, e.g., considering new degrees of freedom that might come into play as the renormalization scale is dialed. Just like for their low energy companions, the contribution of higher energy degrees of freedom to the cosmological constant will be cancelled by the ground state expectation value.

Finally, note that ground states and adiabatic vacua have different coefficients, $\mathcal{B}_{\mathrm{GS}}$ and $\mathcal{B}_{\mathrm{AV}}$, at second adiabatic order. This is due to the local curvature ambiguities in the definition of the stress-energy expectation value as discussed in~\cite{Wald:1995yp}. Therefore, when we subtract the ground state family expectation value, we do not remove the vacuum contribution completely, ${\langle \hat{T}_{\mu\nu} \rangle}^{\mathrm{vac.}} \neq 0$. In this sense, we still allow vacuum fluctuations to play a role in the renormalization of gravity in curved backgrounds.

\section{Examples}\label{sec:Examples}

In this section, we compute the stress-energy tensor expectation value in the instan\-ta\-neous ground states and adiabatic vacua, and confirm the results of the last section in three different models: a scalar, a Dirac spinor and a Proca vector field. Some, but not all, of the results presented in this section are present in the literature \cite{Birrell:1982ix,Kaya_2011,Cherkas:2006kx,Bilic:2010xd}, and we find it convenient to rederive them here. These toy models demonstrate explicitly how the instantaneous ground states provide a covariantly conserved expectation value; how the adiabatic vacua fulfill the same property up to second adiabatic order; and how the subtraction cancels the radiative contributions to the cosmological constant but not those to the gravitational constant.

Regularization is an essential step when dealing with a divergent expectation value. To ensure that the vacuum expectation value of the stress-energy tensor has the correct properties, e.g., it is covariantly conserved, it is important to use a covariant regularization method. For bosonic theories, we use dimensional regularization; for the spin-$1/2$ theory, we use Pauli-Villars regularization\footnote{Note that the divergence in the spin-$1/2$ theory is logarithmic, while it is quadratic in the bosonic theories, therefore we could not reliably use Pauli-Villars regularization for the latter~\cite{Ossola2003}. On the other hand, the generalization of gamma matrices to an arbitrary number of dimensions is a tricky issue, therefore we chose not to use dimensional regularization in the fermionic theory.} and fix the number of spacetime dimensions to 4.

The only dimensionful coupling included in our examples is the gravitational constant, which, in $n$ spacetime dimensions, has mass dimension $[G^{(n)}] = M^{2-n}$. In dimensional regularization, in order to preserve the correct dimensionality while expanding around $n=4$, we introduce $G^{(n)} = G \mu^{4-n}$ where $\mu$ is an arbitrary mass scale and $[G] = M^{-2}$. The second equation in \eqref{ParaReno} becomes
\begin{align}
\label{Renorm2}
G^{\mathrm{ren.}}
=	\frac{G}{1 - 8 \pi G
		\pr{\mathcal{B}_{\mathrm{AV}}
		- \mathcal{B}_{\mathrm{GS}}}
		\mu^{4-n} }
    \;.
\end{align}
For simplicity, we specialize throughout this section to the flat FLRW metric with line element
\begin{align}
\label{Metric-FLRW-ConfFlat}
\dd s^2
=	{a(t)}^2 \pr{\dd t^2
		- \sum_{i=1}^{n-1} \dd x_i^2}
\;,
\end{align}
where $t = x_0$ is the conformal time and $x_i$ are comoving coordinates. We denote derivatives with respect to the conformal time $t$ by a prime $'$.

\subsection{Spin 0}

Consider a scalar field $\phi$ with the action
\begin{align}
S_\phi = \int \dd^n x \; \sqrt{\vrr{g}} \pr{
	\half \, g^{\mu\nu} \, \pd_\mu \phi \, \pd_\nu \phi
	- \half \pr{m^2 - \xi R} \phi^2
	}
\;,
\end{align}
where $m$ is the mass, $\xi$ is a dimensionless coupling constant and $R$ is the Ricci scalar. The equation of motion for $\phi$ is given by
\begin{align}
\phi'' + \pr{n-2} \frac{a'}{a} \, \phi'
- \vec{\pd}^2 \phi + a^2 \pr{m^2 - \xi R} \phi = 0
\;.
\end{align}

\subsubsection*{Decomposition.}

We decompose the field $\phi$ into its Fourier modes,
\begin{align}
\phi(t,x)
&=	\int \frac{\dd^{n-1}k}{\pr{2\pi}^{\pr{n-1}/2}}
	\frac{1}{\sqrt{2}} \, {a(t)}^{-\pr{n-2}/2}
	\nonumber \\ &\hspace{0.9cm} \times
	\pr{\chi_k^{\phantom{*}}(t) \, e^{-ikx} \,
		\hat{a}_k^{\phantom{\dagger}}
		+ \chi_k^*(t) \, e^{ikx} \, \hat{a}_k^\dagger}
\;,
\end{align}
where $\hat{a}_k^{\phantom{\dagger}}$ and $\hat{a}_k^\dagger$ are annihilation and creation operators, and $\chi_k^{\phantom{*}}$ is a complex mode function. Note that this decomposition is not unique: One can choose a different set of annihilation and creation operators or, equivalently, a different set of mode functions. Each set of operators defines a vacuum according to $\hat{a}_k \ket{0}=0$, $\forall k$.

The mode functions satisfy the equation of motion
\begin{align}
\label{Scalar-EOM2}
    \chi_k'' + \pr{k^2 + m^2 a^2 + \pr{\xi_n - \xi} a^2 R} \chi_k = 0
    \;,
\end{align}
where we defined $\xi_n \equiv \frac{n-2}{4\pr{n-1}}$. The scalar field is minimally coupled when $\xi = 0$ and conformally coupled when $\xi = \xi_n$.

We perform the canonical quantization by imposing canonical commutation relations $[ \hat{\phi}(t,x), \hat{\Pi}(t,y) ] = i \, \delta^{n-1}(x-y)$, where $\Pi \equiv \delta S_\phi / \delta \phi' = a^{n-2} \, \phi'$ is the conjugate momentum, as well as $[ \hat{a}_k^{\phantom{\dagger}}, \hat{a}_l^\dagger ] = \delta^{n-1}(k-l)$. Consistency between the commutation relations requires the Wronskian condition
\begin{align}
\label{Scalar-Wronskian}
    \chi_k^{\phantom{*}} \, \chi_k'^*
    - \chi_k' \, \chi_k^{*} = 2 i
    \;.
\end{align}
Since~\eqref{Scalar-EOM2} is a second-order differential equation of a complex function, the space of solutions is 4-dimensional for each mode $k$. One of them is an arbitrary global phase. The Wronskian condition \eqref{Scalar-Wronskian} constrains one more degree of freedom. Then, we are left with 2 \emph{physical} degrees of freedom for $\chi_k$ (for each $k$). 

It is possible to decouple the non-physical degrees of freedom from the mode functions as follows: We write the mode function $\chi_k$ in the polar form as $\chi_k = R_k \, \exp{-i S_k}$, where $R_k(t)$ and $S_k(t)$ are real functions. In these new variables, the Wronskian condition \eqref{Scalar-Wronskian} becomes $R_k^2 \, S_k' = 1$. This equation is solved by any positive real function $\Omega_k$ with $\Omega_k = R_k^{-2}$ and $S_k(t) = \int^t \Omega_k(\bar{t}) \, \dd\bar{t}$ with an arbitrary lower limit of integration,
\begin{align}
\label{Scalar-WKB}
    \chi_k(t) = \frac{1}{\sqrt{\Omega_k(t)}} \,
    \exp{-i \int^t \Omega_k(\bar{t}) \, \dd\bar{t}}
    \;.
\end{align}
Then, the equation of motion \eqref{Scalar-EOM2} becomes
\begin{align}
\label{Scalar-EOM3}
    \Omega_k^2 &= \omega_k^2
    + \pr{\xi_n - \xi} a^2 R
    + \frac{3}{4} \frac{\Omega_k'^2}{\Omega_k^2}
    - \half \frac{\Omega_k''}{\Omega_k}
    \;.
\end{align}
Here, and for the examples to follow, $\omega_k^2 = k^2 + m^2 a^2$.

The solutions $\Omega_k$ of \eqref{Scalar-EOM3} have 2 degrees of freedom (for each $k$), which correspond to the 2 physical degrees of freedom in $\chi_k$. We solved the Wronskian condition and the global phase went into the arbitrary lower integration bound in \eqref{Scalar-WKB}. 
The problem of defining the vacuum has reduced to choosing a solution to the non-linear differential equation \eqref{Scalar-EOM3}.

\subsubsection*{Stress-energy tensor.}

The stress-energy tensor $T_{\mu\nu}$ for the scalar field $\phi$ is given by
\begin{align}
\label{Scalar-SETOp}
    T_{\mu\nu}
    &=	\pd_\mu \phi \, \pd_\nu \phi
	+ \xi R_{\mu\nu} \, \phi^2
	- \xi \, \nabla_\mu \nabla_\nu (\phi^2)
	+ \xi \, g_{\mu\nu} \, \square(\phi^2)
	\nonumber \\ & \hspace{0.5cm}
	- \half \, g_{\mu\nu} \, g^{\alpha\beta} \,
		\pd_\alpha \phi \, \pd_\beta \phi
	+ \half \, g_{\mu\nu} \pr{m^2 - \xi R} \phi^2
	\;.
\end{align}
We promote this to an operator $\hat{T}_{\mu\nu}$ with \emph{symmetrized} operator ordering\footnote{Normal ordering, the standard operator ordering in QFT on flat spacetimes, does not have a generally covariant analogue, and cannot be implemented in our setting. Symmetrized ordering, on the other hand, is covariant.}, e.g., $\phi \phi' \rightarrow \half \{ \hat{\phi} , \hat{\phi}' \}$,~\cite{Bunch}
\subeq{\label{Scalar-SETVEV}\begin{align}
\label{Scalar-Energy}
    \bra{0} \hat{T}_{00} \ket{0} &=
    \frac{1}{4 a^{n-2}}
    \int \frac{\dd^{n-1}k}{\pr{2\pi}^{n-1}} \,
    \frac{1}{\Omega_k}
        \nonumber \\ & \hspace{0.5cm} \times
        \pr{
        k^2 + m^2 a^2
        + \Omega_k^2
        + \frac{\Omega_k'^2}{4\Omega_k^2}
        + \Xi}
\;, \displaybreak[0] \\
    \bra{0} \hat{T}_{0j} \ket{0} &= 0
\;, \displaybreak[0] \\
\label{Scalar-Pressure}
    \bra{0} \hat{T}_{ij} \ket{0} &=
    \frac{1}{4 a^{n-2}}
    \int \frac{\dd^{n-1}k}{\pr{2\pi}^{n-1}} \,
    \frac{1}{\Omega_k}
        \nonumber \\ & \hspace{0.5cm} \times
        \bigg(
        2 k_i k_j
        + \delta_{ij} \pr{\xi_n - \xi} 4 \xi a^2 R
        + \delta_{ij} \, \Xi
        \nonumber \\ & \hspace{0.5cm}
        + \delta_{ij} \pr{1 - 4\xi} \pr{
        \Omega_k^2 - k^2 - m^2 a^2
        + \frac{\Omega_k'^2}{4 \Omega_k^2} }
        \!\bigg)
\;,
\end{align}}
where
\begin{align}
    \Xi &\equiv \pr{\xi_n - \xi} \pr{n-1} \frac{
        a' \pr{\pr{n-2} a' \, \Omega_k
        + 2 a \, \Omega_k'}}{
        a^2 \, \Omega_k}
\end{align}
and $i,j \in \cur{1,...,n-1}$. 

Note that if $\Omega_k$ depends only on the modulus $k \equiv \vert \vec{k} \vert$ of the mode vector $\vec{k}$, then all non-diagonal terms of $\bra{0} \hat{T}_{\mu\nu} \ket{0}$ vanish. This is the case for the adiabatic and ground state vacua that we will consider. The integrals are also simplified by the imposition of spherical symmetry. 

\subsubsection*{Ground state.}

The instantaneous ground state at time $t$ is given by the initial values $\cur{\Omega_k(t),\Omega_k'(t)}$ which minimize the energy density $\bra{0} \hat{T}_{00} \ket{0}(t)$ at that time. Solving the equations $\frac{\pd}{\pd \Omega_k(t)} \bra{0} \hat{T}_{00} \ket{0}(t) = 0 = \frac{\pd}{\pd \Omega_k'(t)} \bra{0} \hat{T}_{00} \ket{0}(t)$ for \eqref{Scalar-Energy}, we find
\subeq{\label{Scalar-GS-IC}\begin{align}
    \Omega_k(t) &=
    \sqrt{k^2 + m^2 \, {a(t)}^2
        + 4 \pr{n-1}^2 \xi \pr{\xi_n - \xi}
            \frac{{a'(t)}^2}{{a(t)}^2}}
    \;,\\
    \Omega_k'(t) &=
    4 \pr{n-1} \pr{\xi - \xi_n} \frac{a'(t)}{a(t)}
        \, \Omega_k(t)
    \;.
\end{align}}
We clarify here that, while these initial conditions can be used together with \eqref{Scalar-EOM3} and \eqref{Scalar-WKB} to define a ground state mode function at all times, we wish to compute our expectation values at every time $t$ on the ground state at that time. The same holds true for the instantaneous adiabatic vacuum.

We therefore substitute the expressions \eqref{Scalar-GS-IC} at every time $t$ into \eqref{Scalar-SETVEV} in order to find the ground state expectation value of the stress-energy tensor at the instant of minimum energy,
\subeq{\label{AdAll-Scalar}\begin{align}
    {\langle \hat{T}_{00} \rangle}_{\mathrm{GS}(t)} &=
    - \frac{a^2}{2} \,
    \frac{1}{\pr{4\pi}^{n/2}} \,
    \Gamma\!\cor{-\frac{n}{2}}
    \nonumber \\ &\hspace{0.5cm} \times
    \pr{m^2 + 4 \pr{n-1}^2 \xi \pr{\xi_n - \xi}
        \frac{a'^2}{a^4}}^{\frac{n}{2}}
    \;, \\
    {\langle \hat{T}_{jj} \rangle}_{\mathrm{GS}(t)} &=
    \frac{a^2}{2} \,
    \frac{1}{\pr{4\pi}^{n/2}} \,
    \Gamma\!\cor{-\frac{n}{2}}
    \nonumber \\ &\hspace{0.5cm} \times
    \pr{m^2 + 4 \pr{n-1}^2 \xi \pr{\xi_n - \xi}
        \frac{a'^2}{a^4}}^{\frac{n}{2}-1}
    \nonumber \\ &\hspace{0.5cm} \times
    \pr{m^2 - 2 \xi \pr{\xi_n - \xi} \mathcal{R} }
    \;,
\end{align}}
where $\mathcal{R} \equiv n R + \pr{n^2 - 2 n + 2} \pr{n-1} \frac{a'^2}{a^4}$. Using these expressions and
\begin{align}
\label{AdAll-CovD}
\nabla^\mu {\langle \hat{T}_{\mu 0} \rangle
	}
&=	\frac{1}{a^2} \, \frac{\pd}{\pd t} \,
        {\langle \hat{T}_{00} \rangle
	    }
	+ \pr{n-3} \frac{a'}{a^3} \,
	    {\langle \hat{T}_{00} \rangle
	    }
    \nonumber \\ &\hspace{0.5cm}
	+ \pr{n-1} \frac{a'}{a^3} \,
	    {\langle \hat{T}_{jj} \rangle
	    }
\;,
    \nonumber \\
\nabla^\mu {\langle \hat{T}_{\mu j} \rangle
	}
&=	0
\;,
\end{align}
for the metric in \eqref{Metric-FLRW-ConfFlat}, one can check that ${\langle \hat{T}_{\mu\nu} \rangle}_{\mathrm{GS}(t)}$ is covariantly conserved.

For the renormalization of $\rho_\Lambda$ and $G$, we expand \eqref{AdAll-Scalar} according to the number of derivatives on the scale factor $a(t)$. We can write the result as
\begin{align}
\label{Scalar-T-GS}
    {\langle \hat{T}_{\mu\nu} \rangle
	}_{\mathrm{GS}(t)} &=
    \pr{\frac{m^2}{4\pi}}^{n/2} \,
    \Gamma\!\cor{-\frac{n}{2}}
    \nonumber \\ &\hspace{0.5cm} \times
    \pr{
        - \half \, g_{\mu\nu}
        + \frac{n}{2} \, \frac{\xi}{\xi_n}
            \pr{\xi - \xi_n} m^{-2} \, G_{\mu\nu} }
\end{align}
up to $\mathcal{O}(\pd^4)$ terms.

\subsubsection*{Adiabatic vacuum.}

In order to define the adiabatic vacua, we try to solve \eqref{Scalar-EOM3} iteratively, such that
\subeq{\begin{align}
    \pr{W_k^{[0]}(t)}^2 &\equiv \omega_k(t)^2
    \;, \\
    \pr{W_k^{[\itero+1]}(t)}^2 &\equiv \omega_k(t)^2
    + \pr{\xi_n - \xi} {a(t)}^2 R(t)
    \nonumber \\ &\hspace{0.5cm}
    + \frac{3}{4} \pr{\frac{W_k^{[\itero]}{}'(t)}{W_k^{[\itero]}(t)}}^2
    - \half \frac{W_k^{[\itero]}{}''(t)}{W_k^{[\itero]}(t)}
    \;.
\end{align}}
The 2nd order adiabatic vacuum at time $t$ are defined by the initial conditions
\begin{align}
\label{Scalar-Adiabatic}
    \Omega_k(t) = W_k^{[1]}(t)
    \;, \qquad
    \Omega_k'(t) = W_k^{[1]}{}'(t)
    \;.
\end{align}
Note that if we had chosen to use a higher order adiabatic vacuum, these equations would change only up to $\mathcal{O}(\pd^4)$. In order to find the adiabatic vacuum expectation value of the stress-energy tensor, we substitute \eqref{Scalar-Adiabatic} into \eqref{Scalar-Energy} and \eqref{Scalar-Pressure}. After expanding the result around the number of derivatives on the scale factor, we get
\begin{align}
\label{Scalar-T-AV}
    {\langle \hat{T}_{\mu\nu} \rangle
	}_{\mathrm{AV}(t)} &=
    \pr{\frac{m^2}{4\pi}}^{n/2} \,
    \Gamma\!\cor{-\frac{n}{2}}
    \nonumber \\ &\hspace{0.5cm} \times
    \pr{
        - \half \, g_{\mu\nu}
        + \frac{n\pr{1-6\xi}}{12} \, m^{-2} \, G_{\mu\nu} }
\end{align}
up to $\mathcal{O}(\pd^4)$ terms.

\subsubsection*{Result.}

Our results \eqref{Scalar-T-GS} and \eqref{Scalar-T-AV} confirm Theorems \ref{thm:magic} and \ref{thm:AV-conservation}, since they are a linear combination of the covariantly conserved tensors $g_{\mu\nu}$ and $G_{\mu\nu}$. They also confirm Theorem \ref{thm:X-property} as
\begin{align}
    \mathcal{A}_\mathrm{GS}
    = \mathcal{A}_\mathrm{AV}
    = - \half \pr{\frac{m^2}{4\pi}}^{n/2} \,
    \Gamma\!\cor{-\frac{n}{2}}
\;.
\end{align}
Finally, by subtracting \eqref{Scalar-T-AV} and \eqref{Scalar-T-GS}, we obtain
\begin{align}
\label{Scalar-BB}
    {\langle \hat{T}_{\mu\nu} \rangle
	}^{\mathrm{vac.}}
	=	\mathcal{B}_\Delta \, G_{\mu\nu}
	    + \mathcal{O}(\pd^4)
\;,
\end{align}
where $\mathcal{B}_\Delta = \mathcal{B}_{\mathrm{AV}} - \mathcal{B}_{\mathrm{GS}}$ is given for the scalar theory by
\begin{align}
    \mathcal{B}_\Delta =
    \pr{\frac{m^2}{4\pi}}^{n/2} \,
    \Gamma\!\cor{-\frac{n}{2}}
    \frac{n}{12} \pr{1 - \frac{24\pr{n-1}}{n-2} \, \xi^2}
    m^{-2}
\;.
\end{align}
The equation \eqref{Scalar-BB} confirms the results of Section \ref{sec:SubtractionScheme}.

In order to complete dimensional renormalization, we expand $\mathcal{B}_\Delta = \mathcal{B}_\Delta(n)$ around $n=4$. Depending on the parameter $\xi$, this can be accomplished in two different ways. If we fix the parameter $\xi$ to the conformal coupling number $\xi = \xi_n$ for every dimension $n$, we get a finite result at $n=4$, namely
\begin{align}
    \mu^{4-n} \, \mathcal{B}_\Delta
    \, \big\vert_{\xi = \xi_n ,\, n=4}
    =   \frac{m^2}{288 \pi^2}
    \;.
\end{align}
Alternatively, we can fix the parameter $\xi$ to a constant independent of $n$ and make a Laurent expansion around $n=4$ to get
\begin{align}
    \mu^{4-n} \, \mathcal{B}_\Delta \big\vert_{\xi = \mathrm{const.}} (n)
    &=   \frac{\pr{1 - 36 \xi^2} m^2}{48 \pi^2 \pr{4-n}}
    \nonumber \\ &\hspace{0.5cm}
    - \frac{\pr{1 - 36 \xi^2} m^2}{96 \pi^2}
        \pr{\gamma + \log{\frac{m^2}{4\pi \mu^2}}}
    \nonumber \\ & \hspace{0.5cm}
    + \frac{\pr{1 - 48 \xi^2} m^2}{96 \pi^2}
    + \mathcal{O}(4-n)
    \;,
\end{align}
where $\gamma$ is the Euler-Mascheroni constant.

\subsection{Spin 1/2}

Next, we consider a massive Dirac fermion $\Psi(t,x)$ in 4 dimensions,
\begin{align}
S_\Psi
    &=	\int \dd^4 x \, \sqrt{-g} \, \Lag_\Psi
\;, \\ \nonumber
\Lag_\Psi &\equiv
		\half \, i \pr{
		\bar{\Psi} \, \gamma^\alpha \,
		\bar{e}_\alpha^{\;\;\mu} \, \nabla_\mu \Psi
		- \pr{\nabla_\mu \bar{\Psi}}
		\bar{e}_\alpha^{\;\;\mu} \, \gamma^\alpha \, \Psi
		}
		- m \bar{\Psi} \Psi
\;.
\end{align}
Here, $e^\alpha_{\;\;\mu} = \operatorname{diag}(a,a,a,a)$ is a tetrad, which satisfies $g_{\mu\nu} = e^\alpha_{\;\;\mu} \, e^\beta_{\;\;\nu} \, \eta_{\alpha\beta}$ for a local flat metric $\eta_{\alpha\beta} = \operatorname{diag}(+,-,-,-)$, while $\bar{e}_\alpha^{\;\;\mu}$ is its inverse. The constant gamma matrices $\cur{\gamma^\alpha}_{\alpha = 0,...,3}$ satisfy $\cur{\gamma^\alpha,\gamma^\beta} = 2 \eta^{\alpha\beta}$ and their curved counterparts are defined as $\tilde{\gamma}^\mu \equiv \bar{e}_\alpha^{\;\;\mu} \, \gamma^\alpha$~\cite{Birrell:1982ix}. We use the Dirac representation
\begin{align}
\gamma^0 = \mtrx{\one & 0 \\ 0 & - \one}
\;, \qquad
\gamma^j = \mtrx{0 & \sigma^j \\ - \sigma^j & 0}
\;.
\end{align}
The conjugate spinor is defined as $\bar{\Psi} = \Psi^\dagger \gamma^0$. The covariant derivative of a Dirac spinor is defined as $\nabla_\mu \Psi =\pd_\mu \Psi + \Gamma_\mu \Psi$, where the connection is given by
\begin{align}
\Gamma_\mu \equiv
	\frac18 \cor{\gamma^\alpha, \gamma^\beta}
	g_{\nu\rho} \, \bar{e}_\alpha^{\;\;\nu}
	\pr{\pd_\mu \bar{e}_\beta^{\;\;\rho}
		+ \Gamma^\rho_{\mu\sigma}
		\bar{e}_\beta^{\;\;\sigma}}
\;,
\end{align}
and $\Gamma^\rho_{\mu\sigma}$ is the Levi-Civita connection. The equation of motion for $\Psi$ is given by
\begin{align}
\label{Fermion-EOM2}
    i \gamma^0 \pr{\frac{1}{a} \, \Psi'
    	+ \frac{3}{2} \, \frac{a'}{a^2} \, \Psi}
    	+ i \, \frac{1}{a} \, \gamma^j \pd_j \Psi
    	- m \Psi = 0
    \;.
\end{align}
Most of the following derivation can already be found, e.g., in \cite{Cherkas:2006kx,Bilic:2010xd}.

\subsubsection*{Decomposition.}

The field operator $\hat{\Psi}$ can be decomposed into a complete set of modes such that
\begin{align}
\label{ModeDecomposition}
\hat{\Psi}(t,x)
&=	{a(t)}^{-3/2} \int \frac{\dd^3k}{\pr{2\pi}^{3/2}}
    \nonumber \\ &\hspace{0.5cm} \times
	\sum_{s = \pm} \pr{
		u^s_k(t) \, e^{-i k x} \, \hat{a}^s_k
		+ v^s_k(t) \, e^{i k x} \, \hat{b}^{s\dagger}_k
		}
\;,
\end{align}
where $s = \pm$ is the spin, $u^s_k$ and $v^s_k$ are the mode functions, $\hat{a}^s_k$ is the particle annihilation operator and $\hat{b}^{s\dagger}_k$ is the anti-particle creation operator. The mode functions satisfy the equations of motion
\subeq{\label{Fermion-EOM3}\begin{align}
\label{Fermion-EOM3u}
    i \gamma^0 {u_k^s}'
	    + \gamma^j k_j \, u_k^s
	    - m a \, u^s_k &= 0
\;, \\
\label{Fermion-EOM3v}
    i \gamma^0 {v_k^s}'
    	- \gamma^j k_j \, v_k^s
    	- m a \, v^s_k &= 0
\;.
\end{align}}
The conjugate momentum $\Pi$ is defined as $\Pi \equiv \delta S_\Psi / \delta \Psi' = i \, a^3 \, \Psi^\dagger$. We perform a canonical quantization of the field operators by imposing the equal-time anti-commutation relations $\{ \hat{\Psi}_a(t,x) , \hat{\Pi}_b(t,y) \} = i \, \delta^3(x-y) \, \delta_{ab}$, as well as $\{ \hat{a}_k^r , \hat{a}_l^{s\dagger} \} = \{ \hat{b}_k^r , \hat{b}_l^{s\dagger} \} = \delta^3(k-l) \, \delta^{rs}$. The requirement for consistency between these anti-commutation relations yields the Wronskian condition,
\begin{align}
\label{Fermion-Wronskian1}
    \sum_{s = \pm} \pr{
    	u^s_k \, u^{s\dagger}_k +
    	v^s_{-k} \, v^{s\dagger}_{-k}}
    =	\one
\;.
\end{align}
The particles and anti-particles are related to each other by charge conjugation. The charge conjugate spinor is defined as $\Psi^c \equiv C \Psi C = - i \pr{\bar{\Psi} \gamma^0 \gamma^2}^T$, together with $C \hat{a}_k^s C = \hat{b}_k^s$ and $C \hat{b}_k^s C = \hat{a}_k^s$. This implies that $u$ and $v$ are related as
\begin{align}
\label{Fermion-Charge-uv}
    u^s_k = - i \gamma^2 \, v^{s*}_k
    \qquad \text{and} \qquad
    v^s_k = - i \gamma^2 \, u^{s*}_k
\;.
\end{align}
Under these relations, one can show that each of \eqref{Fermion-EOM3u} and \eqref{Fermion-EOM3v} is satisfied if and only if the other one is satisfied.

Now, we make the ansatz~\cite{Greene:2000ew}
\begin{align}
    v^s_k = \mtrx{
        \pr{i \chi_k' + m a \, \chi_k} \varphi_s
        \\
        \sigma^j k_j \, \chi_k \, \varphi_s
        } \;,
\end{align}
where the 2-spinors
$\varphi_+ = \pr{1, 0}^T$ and $\varphi_- = \pr{0, 1}^T$
are helicity eigenstates, and $\chi_k(t)$ is a $\C$-valued function that satisfies $\chi_{-k} = \chi_k$. With this ansatz, the equations of motion \eqref{Fermion-EOM3} and the Wronskian condition \eqref{Fermion-Wronskian1} are reduced to the two equations
\begin{align}
\label{Fermion-EOM5}
    \chi_k'' + \pr{k^2 + m^2 a^2 - i m a'} \chi_k
    &= 0
\end{align}
and
\begin{align}
\label{Fermion-Wronskian4}
    \vrr{i \chi_k' + m \, a \, \chi_k}^2
	+ k^2 \vrr{\chi_k}^2 &= 1
\;.
\end{align}
We proceed analogously to the case of a scalar field. We can write the complex mode function $\chi$ in polar form as $\chi_k(t) = R_k(t) \, \exp{- i \int^t \Omega_k(\bar{t}) \, \dd\bar{t}}$ with two real functions $R_k$ and $\Omega_k$, and with an arbitrary lower boundary on the integral that corresponds to the global phase ambiguity. The equation of motion \eqref{Fermion-EOM5} can be split into a real and an imaginary part,
\subeq{\begin{align}
\label{Fermion-EOM6r}
    R_k''
    + \pr{k^2 + m^2 a^2 - \Omega_k^2} R_k
    &= 0
    \;, \\
\label{Fermion-EOM6i}
    2 \, \Omega_k \, R_k'
    + \pr{m \, a' + \Omega_k'} R_k
    &= 0
    \;.
\end{align}}
The second equation \eqref{Fermion-EOM6i} is solved by
\begin{align}
    R_k(t) = \frac{c_1}{\sqrt{\Omega_k(t)}}
    \, \exp{- \half \int^t_{t_1}
        \frac{m a'(\bar{t})}{\Omega_k(\bar{t})} \, \dd\bar{t}}
\;,
\end{align}
where $c_1$ and $t_1$ are arbitrary constants. Then, \eqref{Fermion-EOM6r} and \eqref{Fermion-Wronskian4} become
\begin{align}
\label{Fermion-EOM7}
    \Omega_k^2 &= k^2 + m^2 a^2
    - \frac{m a'' + \Omega_k''}{2 \Omega_k}
    + \frac{\pr{m a' + \Omega_k'}
        \pr{m a' + 3 \Omega_k'}}{4 \Omega_k^2}
\end{align}
and
\begin{align}
\label{Fermion-Wronskian5}
    &
    \frac{c_1^2}{\Omega_k} \,
	\exp{- \int^t_{t_1}
        \frac{m a'(\bar{t})}{\Omega_k(\bar{t})} \, \dd\bar{t}}
    \nonumber \\ & \times
    \pr{k^2 + \pr{m a + \Omega_k}^2
        + \frac{\pr{m a' + \Omega_k'}^2}{4\Omega_k^2}}
    = 1
    \;.
\end{align}
Now, the time derivative of \eqref{Fermion-Wronskian5} is satisfied automatically by virtue of \eqref{Fermion-EOM7}. Hence, we can fix the constant $c_1$ at an arbitrary time, say $t=t_1$, as
\begin{align}
    c_1 = c_1(t_1) &\equiv
        \sqrt{\Omega_k(t_1)}
        \, \Bigg(
        k^2 + \pr{m a(t_1) + \Omega_k(t_1)}^2
        \nonumber \\ &\hspace{2.2cm}
        + \frac{\pr{m a'(t_1) + \Omega_k'(t_1)}^2}
        {4\Omega_k(t_1)^2}
        \Bigg)^{-1/2}
    .
\end{align}
Finally, we have shown that we can write the complex mode function as
\begin{align}
    \chi_k(t) = \frac{c_1(t_1)}{\sqrt{\Omega_k(t)}}
    \, \exp{- i \int^t \Omega_k(\bar{t}) \, \dd\bar{t}
        - \half \int^t_{t_1}
        \frac{m a'(\bar{t})}{\Omega_k(\bar{t})} \, \dd\bar{t}}
    \;,
\end{align}
where $\Omega_k$ is a real (positive) valued function which satisfies \eqref{Fermion-EOM7}. The arbitrariness of the parameter $t_1$ is supported by the observation that the value of $\chi_k$ is independent of $t_1$. 

\subsubsection*{Stress-energy tensor.}

The stress-energy tensor $T_{\mu\nu}$ for the Dirac spinor $\Psi$ is given by~\cite{Birrell:1982ix}
\begin{align}
    T_{\mu\nu}
    &=	\half \, i \, \bar{\Psi} \,
    	\tilde{\gamma}_{(\mu} \nabla_{\nu)} \Psi
    	- \half \, i \pr{\nabla_{(\mu} \bar{\Psi}}
    	\tilde{\gamma}_{\nu)} \, \Psi
\;.
\end{align}
We promote this to an operator $\hat{T}_{\mu\nu}$ with \emph{anti-symmetrized} operator ordering, such as $\Psi \Psi' \rightarrow \half [ \hat{\Psi} , \hat{\Psi}' ]$. After a long but straightforward calculation, we find
\subeq{\begin{align}
\label{Fermion-T00}
    \bra{0} \hat{T}_{00} \ket{0} &=
    \frac{2}{a^2} \int \frac{\dd^3k}{\pr{2\pi}^3} \pr{
        m a + \frac{2 k^2 \, \Omega_k}{K}}
    \;, \\
    \bra{0} \hat{T}_{0j} \ket{0} &= 0
    \;, \\
\label{Fermion-Tij}
    \bra{0} \hat{T}_{ij} \ket{0} &=
    \frac{4}{a^2} \int \frac{\dd^3k}{\pr{2\pi}^3} \,
        \frac{k_i k_j \pr{m a + \Omega_k}}{K}
    \;,
\end{align}}
where
\begin{align}
    K \equiv
    k^2 + \pr{m a + \Omega_k}^2
    + \frac{1}{4\Omega_k^2} \pr{m a' + \Omega_k'}^2
    \;.
\end{align}

\subsubsection*{Ground state.}

Now, we examine the ground state. The initial conditions $\Omega_k(t)$ and $\Omega_k'(t)$ which minimize the energy density in \eqref{Fermion-T00} at time $t$ are given by
\subeq{\begin{align}
\label{Fermion-GS-IC}
    \Omega_k(t) &= \sqrt{k^2 + m^2 \, a(t)^2}
    \;, \\
    \Omega_k'(t) &= - m a'(t)
    \;.
\end{align}}
From these, we get
\subeq{\label{Fermion-T-GS}\begin{align}
    {\langle \hat{T}_{00} \rangle
	}_{\mathrm{GS}(t)} &=
    \frac{1}{\pi^2 a^2}
    \int_0^\infty \dd k \; k^2 \, \omega_k
    \;, \\
    {\langle \hat{T}_{ij} \rangle
	}_{\mathrm{GS}(t)} &=
    \frac{\delta_{ij}}{\pi^2 a^2}
    \int_0^\infty \dd k \; \frac{k^4}{3 \, \omega_k}
    \;.
\end{align}}
These integrals diverge quartically, therefore we cannot use Pauli-Villars regu\-la\-ri\-zation yet. Nevertheless, the covariant conservation of the ground state expectation values can be shown using \eqref{Fermion-T-GS} and the formula \eqref{AdAll-CovD} without evaluating the integrals.

\subsubsection*{Adiabatic vacuum.}

For the adiabatic vacua, we try to solve \eqref{Fermion-EOM7} iteratively,
\subeq{\begin{align}
    \pr{W_k^{[0]}}^2 &\equiv \omega_k^2
    \;, \\
    \pr{W_k^{[\itero+1]}}^2 &\equiv \omega_k^2
    - \frac{m a'' + W_k^{[\itero]}{}''}{2 W_k^{[\itero]}}
    \nonumber \\ &\hspace{0.5cm}
    + \frac{\pr{m a' + W_k^{[\itero]}{}'}
        \pr{m a' + 3 W_k^{[\itero]}{}'}
        }{4 \big( W_k^{[\itero]} \big)^2}
    \;.
\end{align}}
The 2nd order adiabatic vacuum at time $t$ are defined by the initial conditions
\begin{align}
\label{Fermion-AdiabaticFreq}
    \Omega_k(t) = W_k^{[1]}(t)
    \;, \qquad
    \Omega_k'(t) = W_k^{[1]}{}'(t)
    \;.
\end{align}
Writing \eqref{Fermion-AdiabaticFreq} into \eqref{Fermion-T00} and \eqref{Fermion-Tij}, then expanding the integrands by the number of derivatives on the scale factor, we get
\subeq{\label{Fermion-T-AV}\begin{align}
    {\langle \hat{T}_{00} \rangle
	}_{\mathrm{AV}(t)} &=
    \frac{1}{\pi^2 a^2}
    \int_0^\infty \dd k \; k^2 \pr{
        \omega_k + m^2 \, \tau_0(k,m) }
    \;, \\
    {\langle \hat{T}_{ij} \rangle
	}_{\mathrm{AV}(t)} &=
    \frac{\delta_{ij}}{\pi^2 a^2}
    \int_0^\infty \dd k \; k^2 \pr{
        \frac{k^2}{3 \, \omega_k} + m^2 \, \tau_1(k,m) }
\end{align}}
up to $\mathcal{O}(\pd^4)$ terms, where we define
\subeq{\begin{align}
    \tau_0(k,m) &\equiv
    - \frac{k^2 \, a'^2}{8 \pr{k^2 + m^2 \, a^2}^{5/2}}
    \;, \\
    \tau_1(k,m) &\equiv
    - \frac{k^2}{24 \pr{k^2 + m^2 \, a^2}^{7/2}}
    \nonumber \\ &\hspace{0.5cm} \times
    \pr{
        \pr{k^2 + m^2 a^2} \pr{a'^2 - 2 a a''}
        + 5 m^2 a^2 a'^2}
    \;.
\end{align}}

\subsubsection*{Result.}

By subtracting \eqref{Fermion-T-AV} and \eqref{Fermion-T-GS}, we get
\subeq{\begin{align}
    {\langle \hat{T}_{00} \rangle}^{\mathrm{vac.}}
	&=	\frac{m^2}{\pi^2 a^2}
        \int_0^\infty \dd k \; k^2 \, \tau_0(k,m)
        + \mathcal{O}(\pd^4)
    \;, \\
    {\langle \hat{T}_{ij} \rangle}^{\mathrm{vac.}}
	&=	\frac{m^2}{\pi^2 a^2} \, \delta_{ij}
        \int_0^\infty \dd k \; k^2 \, \tau_1(k,m)
        + \mathcal{O}(\pd^4)
    \;.
\end{align}}
These integrals diverge only logarithmically, therefore we can use Pauli-Villars regu\-la\-ri\-zation on them. Since $T_{\mu\nu}$ has mass scale $[M^2]$, we factor out $m^2$ and consider the rest as a function of the mass. For an auxiliary mass scale $\mu$, we find
\subeq{\begin{align}
    &
    \frac{m^2}{\pi^2 a^2}
        \int_0^\infty \dd k \; k^2
        \pr{ \tau_0(k,m) - \tau_0(k,\mu) }
    \nonumber \\
    &=  \frac{m^2 \log\pr{m/\mu}}{24 \pi^2} \,
        \frac{3 a'^2}{a^2}
    \;,
\intertext{and}
    &
    \frac{m^2}{\pi^2 a^2} \, \delta_{ij}
        \int_0^\infty \dd k \; k^2
        \pr{ \tau_1(k,m) - \tau_1(k,\mu) }
    \nonumber \\
    &=  \frac{m^2 \log\pr{m/\mu}}{24 \pi^2} \,
        \frac{a'^2 - 2 a a''}{a^2}
    \;.
\end{align}}
Therefore, in agreement with the general discussion of Section~\ref{sec:SubtractionScheme}, we find that subtracted vacuum fluctuations do not renormalize the cosmological constant,
\begin{align}\splitt{
    {\langle \hat{T}_{\mu\nu} \rangle}^{\mathrm{vac.}} &=
    \mathcal{B}_\Delta(\mu) \, G_{\mu\nu} + \mathcal{O}(\pd^4)
    \;, \\
    \mathcal{B}_\Delta(\mu) &=
    \frac{1}{24 \pi^2} \, m^2 \log{\frac{m}{\mu}}
    \;.
}\end{align}

\subsection{Spin 1}

As our final application, we consider the Proca theory for a massive vector boson $A_\mu(t,x)$ in $n$ dimensions,
\begin{align}
\label{Proca-Action}
    S_A = \int \dd^n x \; \sqrt{\vrr{g}}
	\pr{- \frac14 \, F^{\mu\nu} F_{\mu\nu}
		+ \half \, m^2 A^\mu A_\mu}
	\;,
\end{align}
where $m$ is the mass and $F_{\mu\nu} = \nabla_\mu A_\nu - \nabla_\nu A_\mu$ is the field-strength tensor. The equation of motion for $A_\mu$ is the Proca equation
\begin{align}
\label{Proca-EOM1}
    \nabla^\mu F_{\mu\nu} + m^2 A_\nu = 0
    \;.
\end{align}
Taking the divergence of the Proca equation yields the Lorentz condition
\begin{align}
\label{Proca-EOM1L}
    \nabla^\mu A_\mu = 0
    \;.
\end{align}
The adiabatic regularization method was previously applied to a spin-1 field in the context of Stueckelberg theory in \cite{Frob:2013qsa}.

\subsubsection*{Quantization.}

We decompose the Proca field into its temporal and spatial com\-po\-nents, such that $A_\mu = ( \phi , \vec{A} )$. The Proca theory is a constrained system: There is no kinetic term for $\phi(t,x)$ in the action \eqref{Proca-Action}.

We denote the conjugate momenta to $\phi$ and $\vec{A}$ by $\Pi$ and $\vec{P}$, respectively. The system has two second-class constraints:
\begin{align}
    \Pi = 0
    \;, \qquad
    \vec{\pd} \cdot \vec{P}
    + m^2 \, a^{n-2} \, \phi = 0
    \;.
\end{align}
Applying the Dirac quantization procedure for a constrained system, we find the commutation relations
\subeq{
\label{Proca-Commutators}
\begin{align}
[ \hat{A}_i(t,x) , \hat{P}_j(t,y) ]
&=	i \, \delta_{ij} \, \delta^{n-1}(x-y)
\;, \\
[ \hat{\phi}(t,x) , \hat{A}_j(t,y) ]
&=	i \, m^{-2} \, a^{2-n} \,
	\frac{\pd}{\pd x^j} \, \delta^{n-1}(x-y)
\;,
\end{align}}
and all other commutators vanish.

\subsubsection*{Decomposition.}

We split the spatial vector $\vec{A}$ into its transverse and longi\-tu\-dinal parts as
\begin{align}
    A_j(t,x) = B_j(t,x) + \pd_j C(t,x)
    \quad \text{with} \quad
    \vec{\pd} \cdot \vec{B} = 0
    \;.
\end{align}
The equations of motion \eqref{Proca-EOM1} and the Lorentz condition \eqref{Proca-EOM1L} yield
\subeq{\begin{align}
    \vec{B}'' + (n-4) \, \frac{a'}{a} \, \vec{B}'
	- (\vec{\pd}^2) \, \vec{B} + m^2 a^2 \, \vec{B}
	&= 0
	\;, \\
	\phi' + (n-2) \, \frac{a'}{a} \, \phi
	- \vec{\pd}^2 C &= 0
	\;, \\
	\vec{\pd}^2 C' - \vec{\pd}^2 \phi
	+ m^2 a^2 \phi &= 0
	\;.
\end{align}}
The Proca field $A_\mu$, obeying the Lorentz condition \eqref{Proca-EOM1L}, has $n-1$ independent polarizations: $n-2$ of them are transversal (in $\vec{B}$) and $1$ is longitudinal (in $\phi$ and $C$). Hence, we expand the functions $\phi, \vec{B}, C$ in their Fourier modes as
\subeq{\begin{align}
	B_j(t,x) &=
	\int \frac{\dd^{n-1}k}{\pr{2\pi}^{(n-1)/2}}
	\sum_{r = 1}^{n - 2} \, \varepsilon_j(k,r)
	\nonumber \\ &\hspace{0.5cm} \times
	\pr{\chi_{k,r}^{\phantom{*}}(t) \, e^{-ikx} \, \hat{a}_k^r
		+ \chi_{k,r}^*(t) \, e^{ikx} \, \hat{a}_k^{r\dagger} }
	\;, \\
    \phi(t,x) &=
    \int \frac{\dd^{n-1}k}{\pr{2\pi}^{(n-1)/2}}
	\pr{u_k^{\phantom{*}}(t) \, e^{-ikx} \, \hat{a}_k^0
		+ u_k^*(t) \, e^{ikx} \, \hat{a}_k^{0\dagger} }
	\;, \\
	C(t,x) &=
	\int \frac{\dd^{n-1}k}{\pr{2\pi}^{(n-1)/2}}
	\pr{v_k^{\phantom{*}}(t) \, e^{-ikx} \, \hat{a}_k^0
		+ v_k^*(t) \, e^{ikx} \, \hat{a}_k^{0\dagger} }
	\;.
\end{align}}
Here, the label $r = 1,...,n-2$ stands for the transversal polarizations of the vector boson. The polarization vectors $\vec{\varepsilon}(k,r)$ of the transverse modes satisfy the transversality condition $\sum_{j=1}^{n-1} k_j \, \varepsilon_j(k,r) = 0$, and the normalization conditions $\sum_{j=1}^{n-1} \varepsilon_j(k,r) \, \varepsilon_j(k,r') = \delta_{rr'}$ and $\sum_{r=1}^{n-2} \varepsilon_i(k,r) \, \varepsilon_j(k,r) = \delta_{ij} - \frac{k_i k_j}{k^2}$. The mode functions $\chi_{k,r}, u_k, v_k$ satisfy the equations of motion
\subeq{\begin{align}
    &\chi_{k,r}'' + \pr{n-4} \frac{a'}{a} \, \chi_{k,r}'
    + \omega_k^2 \, \chi_{k,r}^{\phantom{*}} = 0 
    \;, \\
    &u_k'' + \pr{n-2} \frac{a'}{a}  u_k'
    + \pr{\omega_k^2 - (n-2)  \frac{a'^2 - a a''}{a^2}}
        u_k^{\phantom{*}}
    = 0
    ,
\end{align}}
and the constraint
\begin{align}
    v_k^{\phantom{*}} &= - \frac{1}{k^2}
    \pr{ u_k' + (n-2) \, \frac{a'}{a} \, u_k^{\phantom{*}} }
    \;.
\end{align}
We impose the canonical commutation relations $[ \hat{a}_k^0 , \hat{a}_{k'}^{0\dagger} ] = \delta^{n-1}(k-k')$ and $[ \hat{a}_k^r , \hat{a}_{k'}^{r'\dagger} ] = \delta_{rr'} \, \delta^{n-1}(k-k')$ on the annihilation and creation operators. The requirement for consistency between these relations and the commutators \eqref{Proca-Commutators} yields the Wronskian conditions
\begin{align}
    \chi_{k,r}^{\phantom{*}} \, \chi_{k,r}'^*
    - \chi_{k,r}^* \, \chi_{k,r}'
    = i \, a^{4-n}
\end{align} 
and
\begin{align}    
    u_k^{\phantom{*}} \, u_k'^* - u_k^* \, u_k'
    =	i k^2 \, m^{-2} \, a^{2-n} \, ,
\end{align}
Similarly as in the previous examples, we can eliminate the Wronskian conditions by writing the complex mode functions in the polar form,
\subeq{\begin{align}
    \chi_{k,r}(t) &=
    \frac{a(t)^{(4-n)/2}}{\sqrt{2 \Omega_{k,r}(t)}} \,
    \exp\cor{-i \int^t \Omega_{k,r}(\bar{t}) \, \dd\bar{t}}
    \;, \\
    u_k(t) &=
    \sqrt{\frac{k^2}{m^2 a(t)^2}} \;
    \frac{a(t)^{(4-n)/2}}{\sqrt{2 \Theta_k(t)}} \,
    \exp\cor{-i \int^t \Theta_k(\bar{t}) \, \dd\bar{t}}
    \;.
\end{align}}
The positive-valued functions $\Omega_{k,r}$ and $\Theta_k$ satisfy the vacuum equations
\subeq{
\label{Proca-EOM4}
\begin{align}
    \Omega_{k,r}^2 &= \omega_k^2
	- (n-4) \, \frac{(n-6) \, a'^2 + 2 a a''}{4a^2}
	+ \frac34 \, \frac{\Omega_{k,r}'^2}{\Omega_{k,r}^2}
	- \half \, \frac{\Omega_{k,r}''}{\Omega_{k,r}^{\phantom{*}}}
	\;, \\
	\Theta_k^2 &= \omega_k^2
	- (n-2) \, \frac{n \;\! a'^2 - 2 a a''}{4a^2}
	+ \frac34 \, \frac{\Theta_k'^2}{\Theta_k^2}
	- \half \, \frac{\Theta_k''}{\Theta_k}
	\;.
\end{align}}

\subsubsection*{Stress-energy tensor.}

The stress-energy tensor $T_{\mu\nu}$ for the Proca field $A_\mu$ is given by
\begin{align}
    T_{\mu\nu} &=
	\frac14 \, g_{\mu\nu} F^{\rho\sigma} F_{\rho\sigma}
	- g^{\rho\sigma} F_{\mu\rho} F_{\nu\sigma}
	\nonumber \\ &\hspace{0.5cm}
	+ m^2 A_\mu A_\nu
	- \half \, m^2 g_{\mu\nu} A^\rho A_\rho
    \;.
\end{align}
Again, we promote this to an operator $\hat{T}_{\mu\nu}$ with \emph{symmetrized} ordering. In a straightforward calculation, we find the vacuum expectation value of this operator to be
\subeq{\begin{align}
\label{Proca-T00}
    \bra{0} \hat{T}_{00} \ket{0} &=
    \frac{1}{4 a^{n-2}}
    \int \frac{\dd^{n-1}k}{\pr{2\pi}^{n-1}} \,
\nonumber \\ &\hspace{0.5cm} \times
    \Bigg\{
    \frac{1}{\Theta_k} \pr{
        \omega_k^2 + \Upsilon_\Theta }
    + \sum_{r=1}^{n-2}
    \frac{1}{\Omega_{k,r}} \pr{\omega_k^2 + \Upsilon_\Omega }
    \Bigg\} \;,
    \displaybreak[0] \\
    \bra{0} \hat{T}_{0j} \ket{0} &= 0 \;,
    \displaybreak[0] \\
    \bra{0} \hat{T}_{ij} \ket{0} &=
    \frac{1}{4 a^{n-2}}
    \int \frac{\dd^{n-1}k}{\pr{2\pi}^{n-1}} \, \Bigg\{
    2 k_i k_j \pr{\frac{1}{\Theta_k}
    + \sum_{r=1}^{n-2} \frac{1}{\Omega_{k,r}} }
    \nonumber \\ & \hspace{0.5cm}
    + \sum_{r=1}^{n-2}
    \frac{1}{\Omega_{k,r}}
    \pr{ 2 \varepsilon_{i,k,r} \, \varepsilon_{j,k,r} - \delta_{ij}}
    \pr{\omega_k^2 - \Upsilon_\Omega }
    \nonumber \\ & \hspace{0.5cm}
    + \frac{1}{\Theta_k}
    \pr{\delta_{ij} - \frac{2 k_i k_j}{k^2}}
    \pr{\omega_k^2 - \Upsilon_\Theta }
    \Bigg\}
\;,
\end{align}}
where
\subeq{\begin{align}
    \Upsilon_\Omega &\equiv
    \Omega_{k,r}^2
    + \frac14 \pr{\frac{\Omega_{k,r}'}{\Omega_{k,r}}
        + \pr{n-4} \frac{a'}{a}}^2
    \;, \\
    \Upsilon_\Theta &\equiv
    \Theta_k^2
    + \frac14 \pr{\frac{\Theta_k'}{\Theta_k}
        + \pr{2-n} \frac{a'}{a}}^2
    \;.
\end{align}}

\subsubsection*{Ground state.}

At any given time $t$, the instantaneous ground state which minimizes the energy density \eqref{Proca-T00} is given by the initial conditions
\subeq{\label{Proca-GS-IC}\begin{align}
    \Omega_{k,r}(t) &= \omega_k(t)
    \;, \quad
    \Omega'_{k,r}(t) =
    \pr{4-n} \frac{a'(t)}{a(t)} \, \omega_k(t)
    \;, \\
    \Theta_k(t) &= \omega_k(t)
    \;, \quad \hspace{6pt}
    \Theta'_k(t) =
    \pr{n-2} \frac{a'(t)}{a(t)} \, \omega_k(t)
    \;.
\end{align}}
Substituting these into the vacuum expectation value of the stress-energy tensor, we get
\begin{align}
\label{Proca-T-GS}
    {\langle \hat{T}_{\mu\nu} \rangle
	}_{\mathrm{GS}(t)} &=
    \pr{\frac{m^2}{4\pi}}^{n/2} \,
    \Gamma\!\cor{1 - \frac{n}{2}}
    \frac{n-1}{n} \, g_{\mu\nu}
    \;.
\end{align}
Since this quantity is proportional to the metric, its covariant conservation is manifest.

\subsubsection*{Adiabatic vacuum.}

As usual, we try to solve \eqref{Proca-EOM4} iteratively,
\subeq{\begin{align}
    \pr{W_k^{[0]}}^2 &\equiv \omega_k^2
    \;, \\
    \pr{W_k^{[\itero+1]}}^2 &\equiv \omega_k^2
	- (n-4) \, \frac{(n-6) \, a'^2 + 2 a a''}{4a^2}
    \nonumber \\ &\hspace{0.5cm}
    + \frac{3}{4} \pr{\frac{W_k^{[\itero]}{}'}{W_k^{[\itero]}}}^2
    - \half \frac{W_k^{[\itero]}{}''}{W_k^{[\itero]}}
    \;, \\
    \pr{Z_k^{[0]}}^2 &\equiv \omega_k^2
    \;, \\
    \pr{Z_k^{[\itero+1]}}^2 &\equiv \omega_k^2
	- (n-2) \, \frac{n \;\! a'^2 - 2 a a''}{4a^2}
    \nonumber \\ &\hspace{0.5cm}
    + \frac{3}{4} \pr{\frac{Z_k^{[\itero]}{}'}{Z_k^{[\itero]}}}^2
    - \half \frac{Z_k^{[\itero]}{}''}{Z_k^{[\itero]}}
    \;.
\end{align}}
The 2nd order adiabatic vacuum at time $t$ are defined by the initial conditions
\subeq{\label{Proca-AdiabaticFreq}\begin{align}
    \Omega_{k,r}(t) &= W_k^{[1]}(t)
    \;, \qquad
    \Omega'_{k,r}(t) = W_k^{[1]}{}'(t)
    \;, \\[5pt]
    \Theta_k(t) &= Z_k^{[1]}(t)
    \;, \qquad \hspace{10pt}
    \Theta'_k(t) = Z_k^{[1]}{}'(t)
    \;.
\end{align}}
The expectation value of the stress-energy tensor is then, to second adiabatic order
\begin{align}
\label{Proca-T-AV}
    {\langle \hat{T}_{\mu\nu} \rangle
	}_{\mathrm{AV}(t)} &=
    \pr{\frac{m^2}{4\pi}}^{n/2} \,
    \Gamma\!\cor{1 - \frac{n}{2}}
    \nonumber \\ &\hspace{0.5cm} \times
    \pr{
        \frac{n-1}{n} \, g_{\mu\nu}
        + \frac{7-n}{6} \, m^{-2} \, G_{\mu\nu} }
    \,
\end{align}
up to $\mathcal{O}(\pd^4)$ terms.
To the best of our knowledge, this is the first derivation of the adiabatic vacuum and adiabatic stress-energy tensor expectation value for a Proca field. 

\subsubsection*{Result.}

Subtracting \eqref{Proca-T-AV} and \eqref{Proca-T-GS}, we finally arrive at
\begin{align}\splitt{
    {\langle \hat{T}_{\mu\nu} \rangle}^{\mathrm{vac.}} &=
    \mathcal{B}_\Delta \, G_{\mu\nu}
    + \mathcal{O}(\pd^4)
    \;, \\
    \mathcal{B}_\Delta &=
    \pr{\frac{m^2}{4\pi}}^{n/2}
    \Gamma\!\cor{1 - \frac{n}{2}} \,
    \frac{7-n}{6} m^{-2}
\;,
}\end{align}
once again confirming the results of Section~\ref{sec:SubtractionScheme}. For dimensional renormalization, we can separate the divergent and finite parts of $\mu^{4-n} \, \mathcal{B}_\Delta$ using a Laurent expansion around $n=4$,
\begin{align}
    \mu^{4-n} \, \mathcal{B}_\Delta &=
    \frac{m^2}{16 \pi^2 \pr{n-4}}
    + \frac{m^2}{32 \pi^2} \pr{\gamma + \log{\frac{m^2}{4\pi \mu^2}}}
    \nonumber \\ &\hspace{0.5cm}
    - \frac{5 m^2}{96 \pi^2}
    + \mathcal{O}(n-4)
    \;.
\end{align}

\section{Conclusion}
We here considered the possibility that the source of gravity in the semiclassical Einstein equation is the difference in stress-energy expectation value between the state of the Universe and the instantaneous ground state. We proved that for homogeneous and isotropic cosmological spacetimes, for which one can of course identify the instantaneous ground state, the proposed stress-energy tensor satisfies the Bianchi identity.

We discussed renormalization and the r\^ole of the physical (no particle, adiabatic) vacuum in our scheme. We find that, as a consequence of the instantaneous ground state subtraction, the vacuum energy density $\rho_\Lambda = \frac{1}{8\pi G} \, \Lambda$ becomes radiatively stable, i.e., protected from renormalization at the UV scale. We demonstrate this explicitly in the case of a scalar field with arbitrary coupling to the scalar curvature, for a free spinor and a free Proca field.

Our assumptions in this study were 1) that the background metric is of the FLRW kind, 2) the existence of instantaneous ground states, and 3) the existence of adiabatic vacua. The validity of the last two assumptions is nontrivial for generic interacting theories. This invites further investigation of the stabilization of the cosmological constant along the lines outlined here.

Finally, it will be very interesting to explore to what extent our results can be developed beyond highly symmetrical cosmological backgrounds. It will also be intriguing to understand if and how existing quantum gravity models can accommodate a first principles derivation for our proposal.

$$$$

\section*{Acknowledgements}

We are thankful to Niayesh Afshordi, Lee Smolin and Neil Turok for their comments on the draft. 
Research at Perimeter Institute is supported in part by the Government of Canada through the Department of Innovation, Science and Economic Development Canada and by the Province of Ontario through the Ministry of Economic Development, Job Creation and Trade. AK acknowledges support through the Discovery Grants Program of the National Science and Engineering Research Council of Canada (NSERC). This research was also partly supported by grants from John Templeton Foundation and FQXi.

\bibliography{CosmologicalConstant.bib}

\end{document}